\documentclass{lmcs}
\pdfoutput=1

\usepackage{lastpage}
\lmcsdoi{21}{2}{28}
\lmcsheading{}{\pageref{LastPage}}{}{}%
{Aug.~12,~2024}{Jun.~24,~2025}{}

\usepackage[utf8]{inputenc}
\usepackage{algpseudocode}
\usepackage{amssymb}
\usepackage{nicefrac}
\usepackage{tikz}
\usetikzlibrary{automata,positioning}
\usepackage{pgfplots}
\usepackage{booktabs}
\pgfplotsset{width=10cm,compat=1.9}

\renewcommand{\vec}[1]{\boldsymbol{#1}}
\newcommand{\supp}[1]{\mathrm{supp}(#1)}
\newcommand{\norm}[2]{\left\lVert #2 \right\rVert_{#1}}
\newcommand{\restr}[2]{\left.#1\right|_{#2}}
\newcommand{\wit}[2]{\mathbf{Wit}(#1,#2)}
\DeclareMathOperator{\Ex}{\mathbb{E}}
\newcommand{\lexleq}{\leq_{\mathrm{lex}}}
\newcommand{\adj}{\mathrm{adj}}

\newcommand\xqed[1]{%
  \leavevmode\unskip\penalty9999 \hbox{}\nobreak\hfill
  \quad\hbox{#1}}
\newcommand\eoe{\xqed{$\triangle$}}

\usepackage{soul}
\usepackage{xcolor}
\DeclareRobustCommand{\revone}[2]{#2}
\DeclareRobustCommand{\revtwo}[2]{#2}

\begin{document}

\title{Algorithms for Markov Binomial Chains}
\thanks{This work was supported by the Flemish inter-university (iBOF) ``DESCARTES'' project.}
\author[A. Alarc\'on Gonzalez]{Alejandro Alarc\'on Gonzalez\lmcsorcid{0000-0002-1959-3307}}
\author[N. Hens]{Niel Hens\lmcsorcid{0000-0003-1881-0637}}
\author[T. Leys]{Tim Leys\lmcsorcid{0000-0001-6281-9693}}
\author[G. A. P\'erez]{Guillermo A. P\'erez\lmcsorcid{0000-0002-1200-4952}}
\address{University of Antwerp, Belgium} \email{guillermo.perez@uantwerpen.be}

\begin{abstract}
    We study algorithms to analyze a particular class of Markov population processes that is often used in epidemiology. More specifically, Markov binomial chains are the model that arises from stochastic time-discretizations of classical compartmental models. In this work we formalize this class of Markov population processes and focus on the problem of computing the expected time to termination in a given such model. Our theoretical contributions include proving that Markov binomial chains whose flow of individuals through compartments is acyclic almost surely terminate. We give a PSPACE algorithm for the problem of approximating the time to termination and a direct algorithm for the exact problem in the Blum-Shub-Smale model of computation. Finally, we provide a natural encoding of Markov binomial chains into a common input language for probabilistic model checkers. We implemented the latter encoding and present some initial empirical results showcasing what formal methods can do for practicing epidemiologists.
\end{abstract}

\maketitle

\section{Introduction}
We study a class of discrete-time \emph{Markov population processes} that is often used to model epidemics. In general, \emph{Markov population models} refers to Markov models whose state space is a discrete partitioning of a population into \emph{colonies} or \emph{compartments}. In other words, their set of states $S$ is such that $S \subseteq \mathbb{N}^k$ for some $k > 0$. Such models arise in the theories of epidemics, population dynamics, 
rumors, systems biology, and queuing and chemical reaction networks (see, for instance,~\cite{Kingman_1969,DBLP:journals/ijfcs/HenzingerJW11,DBLP:journals/tcs/CairoliAdB23} and references therein).

Due to their various important applications, the literature around Markov population models covers interesting mathematical properties of various subclasses. Importantly, and perhaps because of their interest to the systems biology community, there is also a wealth of formal methods to analyze them (see, e.g., \cite{DBLP:journals/nla/DayarHSW11,DBLP:conf/hybrid/LapinMW11,DBLP:journals/iandc/BortolussiLN18,DBLP:conf/qest/BackenkohlerBW20,DBLP:journals/tcs/CairoliAdB23}). Importantly, encodings of Markov population models into related stochastic models and language formats accepted by formal-verification tools are known~\cite{DBLP:journals/ijfcs/HenzingerJW11}. \revone{R1.5}{In addition, techniques from fluid limits can be used to efficiently infer stochastic information of interest from some of the continuous-time classes of such models}~\cite{DBLP:conf/concur/BortolussiH12,DBLP:conf/asmta/Bortolussi10}.

Most of the literature concerning Markov population processes focuses on continuous-time variants of them. This is in contrast with the class of processes that we study. \emph{(Markov) Binomial chains}~\cite{bailey1975mathematical} are discrete-time Markov population processes whose transition probabilities are given by a product of probability mass functions of binomial distributions over possible individual transfers between compartments. \revone{R1.1}{While this class of model was initially conceived for simple scenarios with small populations, it has recently been used for more complex situations like the analysis of COVID-19 cases}~\cite{ABRAMS2021100449,prunas2022vaccination}. \revone{R1.1}{Recently, it has also been established that (for epidemiological tasks) discrete-time models are as general and as flexible as their continuous-time counterparts, yet they are simpler to parameterize on the basis of data and to implement computationally}~\cite{diekmann2021discrete}.

We are unaware of other studies of algorithms and complexity-theoretic results for binomial chains as defined in this work. Related discrete-time stochastic models that have been studied thoroughly include branching processes~\cite{DBLP:journals/siamcomp/EtessamiSY17,DBLP:journals/iandc/EtessamiSY18,DBLP:journals/mor/EtessamiSY20} and probabilistic vector addition systems~\cite{DBLP:conf/lics/BrazdilKKN15,DBLP:conf/concur/Ajdarow023}. Also worth mentioning is an algorithm by Black and Ross to compute the final (population-)distribution for another class of discrete-time Markov population protocols~\cite{BLACK2015159}.

\paragraph{Contributions.} In this work, we initiate the study of algorithms and formal methods for the analysis of binomial chains. We start by giving a self-contained account of how classical compartmental models give rise to binomial chains, in \autoref{sec:sir}. \revone{R1.1}{The derivation of a binomial chain from compartmental models in that section is meant as a way to motivate the class of models and to provide useful epidemiological context and intuition to the unfamiliar reader. Importantly, we do not mean for the translation to be used as a way to analyze compartmental models by approximating them with a binomial chain. As per the motivation above, our target use case is when the ground-truth model is already a binomial chain fitted to historical data by epidemiologists.} In  in \autoref{sec:definitions}, we then move to formalizing the general model of binomial chains.
After that, in \autoref{sec:time2term-acyclic}, we prove that binomial chains whose flow of individuals through compartments is acyclic almost surely terminate. Finally, in \autoref{sec:exhit-approx}, we give a \textbf{PSPACE} algorithm to approximate the time to termination in a given binomial chain and, in \autoref{sec:sir-algo}, we also give a direct algorithm for the exact problem (ignoring issues with irrational numbers and the complexity of arithmetic operations). To close the paper, we give an encoding of binomial chains into a common input language for probabilistic model checkers in \autoref{sec:encoding} and, in \autoref{sec:experiments}, we present some empirical results obtained with an implementation of this encoding.

\section{Preliminaries}
We write $\mathbb{N}$ for the set of all natural numbers, including $0$; $\mathbb{Q}$ and $\mathbb{Q}_{\geq 0}$, for the sets of all rational and nonnegative rational numbers, respectively; and $\mathbb{R}$ and $\mathbb{R}_{\geq 0}$, for the sets of all real and nonnegative real numbers, respectively. Also, we use $\lg(x)$ to denote the logarithm of $x$ with respect to base $2$; and $\ln(x)$ for the natural logarithm of $x$, i.e. with respect to base $e$. Instead of $e^x$, we sometimes write $\exp(x)$. 

Let $k \in \mathbb{N}$ be such that $k \geq 1$. We write $[n]$ for the set $\{1, 2, \dots, n\}$.
For vectors $\vec{u},\vec{v} \in \mathbb{Q}^{k}$ with rational entries, we write $\vec{u} \leq \vec{v}$ to denote that $u_i \leq v_i$ for all $1 \leq i \leq k$. Let $\ell \in \mathbb{N}$ be such that $\ell \geq 1$ and $\vec{M} \in \mathbb{Q}^{k \times \ell}$ be a matrix with rational entries. We write $\supp{\vec{M}}$ to denote the support of $\vec{M}$, that is, the set of indices such that the corresponding entry of $\vec{M}$ is nonzero. In symbols, 
\[
\supp{\vec{M}} = \{(i,j) \in [k] \times [\ell] : M_{ij} \neq 0\}.
\]
We extend the notion of support to vectors $\vec{u}$ in the natural way and write $i \in \supp{\vec{u}}$ instead of $(i,1) \in \supp{\vec{u}}$.

It will also be convenient to have notation for a class of linear transformations from vectors of natural numbers to nonnegative scalars. Let $k \in \mathbb{N}$ be such that $k \geq 1$. We write $\mathbb{L}_k$ for the set of all functions $f : \mathbb{N}^k \to \mathbb{Q}_{\geq 0}$ such that there exists a vector $\vec{a} \in \mathbb{Q}^k_{\geq 0}$ and a scalar $b \in \mathbb{Q}_{\geq 0}$ for which $f(\vec{x}) = \vec{a}^\intercal \vec{x} + b$.

Finally, to make sure our notation for them is clear, we give a definition for Markov chains. A discrete-time Markov chain is a tuple $(S,s_0,P)$ where $S$ is a countable set of states, $s_0 \in S$ is the initial state, and $P : S \times S \to \mathbb{R}_{\geq 0}$ is a transition probability function, that is, for all $s \in S$ we have that $\sum_{s' \in S} P(s,s') = 1$. We write $X_t$, where $t \in \mathbb{N}$, for the random variable representing the state of the chain at the $(t+1)$-th step and $\Pr(X_t = s)$ for the probability measure of \emph{runs} (i.e. sequences of transitions) of the chain with $s$ as their $(t+1)$-th state. In a slight abuse of notation, whenever $S$ is finite, we sometimes write $\vec{P}$ for the matrix with entries $P_{ij} = P(i,j)$, for all $i,j \in S$.

\section{SIR Models: From ODEs to a bimonial chain}\label{sec:sir}

Deterministic SIR Models are simple mathematical models of the spreading of infectious
diseases. In them, a \emph{population} of size $N \in \mathbb{N}$ is partitioned into
\emph{compartments} with labels: $S$ for susceptible, $I$ for infectious, and
$R$ for recovered. As illustrated in \autoref{fig:flow_diagram}, people may move between compartments following time-dependent
dynamics which are usually prescribed by ordinary differential equations (ODEs).
\begin{align*}
  \dfrac{dS(t)}{dt} &{}= - \beta I(t) S(t)\\
  \dfrac{dI(t)}{dt} &{}= \beta I(t) S(t) - \gamma I(t)\\
  \dfrac{dR(t)}{dt} &{}= \gamma I(t)
\end{align*}
Hence, we
write $S(t)$, $I(t)$, and $R(t)$ to highlight the fact that these values are
functions of time $t \in \mathbb{R}_{\geq 0}$. Below, we first state the deterministic SIR
model for the case of a closed population~\cite{ModellingInfectiousDiseases}, i.e., births, deaths or infections resulting from contacts with individuals from outside the population are not being considered. (The values $\beta$ and $\gamma$ are
explained in the sequel.) Then, we derive a stochastic discrete-time version
thereof, based on Bailey's chain binomial~\cite{bailey1975mathematical}. 

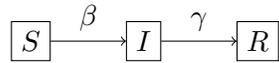
\begin{figure}
  \centering
  \begin{tikzpicture}
    \node[rectangle,draw](s){$S$};
    \node[rectangle,draw,right= of s](i){$I$};
    \node[rectangle,draw,right= of i](r){$R$};
    \path[->,auto]
      (s) edge node{$\beta$} (i)
      (i) edge node{$\gamma$} (r);
  \end{tikzpicture}
\caption{Overview of the flow of individuals in the SIR model:
Following a disease infection, susceptible individuals ($S$) move to an
infectious state ($I$) in which they can infect others. Such infectious
individuals will recover over time.}
\label{fig:flow_diagram}
\end{figure}

\subsection{Towards a stochastic SIR model}
When an infectious individual makes contact with a susceptible individual,
there is some probability that such contact will lead to
disease transmission. This probability, multiplied by the \emph{contact rate}, is denoted by $\beta^*$, and we take it
to be irrespective of the specific susceptible-infectious pair. Furthermore,
we define the \textit{force of infection}, denoted by $\Lambda(t)$, as the
rate for a susceptible individual to become infected
at time t. Assuming \textit{homogeneous mixing} within the population yields the
relation
\begin{equation}
\label{eq:lawofmassaction}
  \Lambda(t)=\dfrac{I(t)}{N}\beta^*.
\end{equation}
To simplify notation, we let $\beta=\beta^*/N$. The formulation of the force of infection in \eqref{eq:lawofmassaction} is referred to as \textit{mass action transmission}~\cite{ModellingInfectiousDiseases}. The exposition above is extended to all susceptible individuals, leading to the continuous-time relation:
\begin{equation}
\label{eq:diffeq}
\dfrac{dS(t)}{dt}=-\Lambda(t)S(t) = -\beta I(t) S(t).
\end{equation}

We start the discretization by fixing $h>0$. By integrating
\eqref{eq:diffeq} over the time interval $(t,t+h]$, the following recurrence relation
is deduced.
\begin{equation} 
\label{eq:expectedinS}
        S(t+h)=e^{-\int_{t}^{t+h} \Lambda(\tau)d\tau}S(t) 
\end{equation}

Equation \eqref{eq:expectedinS} is now interpreted as the expected number of
susceptible individuals at time $t+h$, assuming there are $S(t)$ susceptible individuals
at time $t$. We observe in turn that the first factor on the right hand side of
\eqref{eq:expectedinS} is the probability for a susceptible to escape from
infection during the time interval $(t,t+h]$. Therefore,
$1-\exp(-\int_{t}^{t+h} \Lambda(\tau)d\tau)$ is the probability for a
susceptible person to become infected during the time interval $(t,t+h]$. 

The following integral becomes the \emph{cumulative force of infection} over $(t,t+h]$.
\begin{equation}
\label{eq:cumulativeforce}
  \int_{t}^{t+h} \Lambda(\tau)d\tau  
\end{equation}
By taking the
Taylor expansion of \eqref{eq:cumulativeforce} around $t$ and considering
expressions \eqref{eq:lawofmassaction}, \eqref{eq:expectedinS} together, the
following holds.
\begin{align*} 
  S(t+h)&{}=e^{-\int_{t}^{t+h} \beta I(\tau)d\tau}S(t)\\
        &{}=e^{-\beta(hI(t)+O( h^2))}S(t)\\
        &{}\approx e^{-h\beta I(t)}S(t),\text{ for $h$ small}   
\end{align*}

The probability $p_1(t)=1-e^{-h\beta I(t)}$ will now be regarded as the
success probability for the Bernoulli trial corresponding to an interaction
between an infectious and a susceptible individual, and the interaction occurs
within $(t,t+h]$. By recalling the assumption of homogeneous mixing, the
previous discussion suggests we define a random variable $I_{t+h}^{new}$ that
follows a binomial distribution and which represents the newly
infected individuals.
\begin{equation}
\label{eq:CBI}
I_{t+h}^{new} \sim B(S(t), p_1(t)=1-\exp(-h\beta I(t)))
\end{equation}

To finish this modeling part, the infectious period $\Delta(t)$ is assumed to
be exponentially distributed with parameter $\gamma \in \mathbb{R}_{\geq 0}$,
irrespective of the individual. Accordingly, the cumulative density function
$p_2=1-\exp(-h\gamma)$ will represent the probability for an infectious
individual to have recovered by time $h>0$. This suggests that the
corresponding number of newly recovered individuals $R_{t+h}^{new}$ can
be defined as a random variable with binomial distribution as follows.
\begin{equation} 
    \label{eq:CBR}
    R_{t+h}^{new} \sim B(I(t), p_2=1-\exp(-h\gamma))
\end{equation}

\subsection{The SIR process}
Let $S(0)$, $I(0)$, and $R(0)$ be fixed constants such that $N = S(0) + I(0) +
R(0)$ --- these are just the initial conditions for the ODE version of an SIR
model --- and $h > 0$. For all $t \in \mathbb{N}$, we define discrete random
variables $S_t$, $I_t$, and $R_t$, all of which take values from $\mathbb{N}$.
In particular, let $S_0 = S(0)$, $I_0 = I(0)$, and $R_0 = R(0)$. For $t \geq
0$ we base the following definition on Equations~\eqref{eq:CBI}
and~\eqref{eq:CBR}:
\begin{align}\label{eq:difference_equations}
    S_{t+1} = {} & S_{t} - Y_{t}\\
    I_{t+1} = {} & I_{t} + Y_{t} - Z_{t}\\
    R_{t+1} = {} & R_{t} + Z_{t} 
\end{align}
with $Y_{t} \sim B(S_t, 1-\exp(-h\beta I_t))$ and $Z_{t} \sim B(I_t, 1-\exp(-h\gamma))$.

By definition, we have the property of \emph{conservation of population}.
\begin{lem}\label{lem:conservation}
  For all $t \in \mathbb{N}$ we have that $N = S_t + I_t + R_t$.
\end{lem}

Let $t \in \mathbb{N}$ be arbitrary and write $(S_t,I_t,R_t) =
(m_1,m_2,m_3)$. We will focus on the probability mass function $p^{(t+1)}_{S}$
of $S_{t+1}$. \revtwo{R2.4}{The equations below follow from our definition.}
\begin{equation} \label{eq:Spmf}
\begin{aligned}
  p^{(t+1)}_S(n_1)={}&\Pr(S_{t+1}=n_1)\\
  ={} & \Pr(S_{t}-Y_{t}=n_1) \\
  ={} & \Pr(Y_{t}=m_1-n_1) \\
  ={} & \binom{m_1}{m_1-n_1}(\exp(-h \beta m_2))^{n_1}(1-\exp(-h \beta m_2)^{m_1-n_1}
\end{aligned}
\end{equation}
(Here, we again adopt the convention that $0^0 = 1$ so that $p^{(t+1)}_S(m_1)$ to be
$1$ when $m_2 = 0$.)
Similarly, for the probability mass function $p^{(t+1)}_{R}$ of $R_{t+1}$, we
get \revtwo{R2.4}{the following equations.}
\begin{equation} \label{eq:Rpmf} 
\begin{aligned}
  p^{(t+1)}_R(n_3)={}&\Pr(R_{t+1}=n_3)\\    
  ={} &\Pr(R_{t} + Z_{t}=n_3) \\
  ={} &\Pr(Z_{t}=n_3-m_3) \\
  ={} & \binom{m_2}{n_3-m_3}(1-\exp(-h\gamma))^{n_3-m_3}\exp(-h\gamma)^{m_2-n_3+m_3}      
\end{aligned}
\end{equation}
Observe, from \eqref{eq:Spmf} and \eqref{eq:Rpmf}, that $S_{t+1}$ and
$R_{t+1}$ are independent. \revtwo{R2.5}{That is, the events $S_{t+1}=n_1$ and $R_{t+1} = n_3$ are independent if we condition on $(S_t,I_t,R_t) = (m_1,m_2,m_3)$. Since $I_{t+1} = N - S_{t+1} - R_{t+1}$,} by
\autoref{lem:conservation}, the SIR process satisfies the Markov property. It remains to compute the exact probability of such a transition.

Recall that if $X$ and $Y$ are independent then:
\[\Pr(X= x,Y= y)=\Pr(X= x)\Pr(Y= y)=p_{X}(x)p_{Y}(y),\]
for $x,y$ in the sample
space. Thus, the joint probability mass function $p^{(t+1)}_{S,R}$ of
$(S_{t+1},R_{t+1})$ is given by $p^{(t+1)}_{S,R} = p^{(t+1)}_{S}
p^{(t+1)}_{R}$. According to \autoref{lem:conservation}, $I_{t+1} = N - S_{t+1} - R_{t+1}$, so the joint probability mass function
$p^{(t+1)}_{S,I,R}$ of $(S_{t+1},I_{t+1},R_{t+1})$ is also given by
$p^{(t+1)}_{S,I,R} = p^{(t+1)}_{S} p^{(t+1)}_{R}$.

\subsection{A discrete-time Markov chain induced by the binomial chain}
Henceforth, we write $X_t$ for $(S_t,I_t,R_t)$. For
all $t \in \mathbb{N}$, $X_t$ is a discrete random variable whose states are vectors
from $(n_1,n_2,n_3) \in \mathbb{N}^3$ such that $N = n_1 + n_2 + n_3$ and its probability mass function (as discussed above),
is given, for all $t \geq 1$, by:
\begin{equation} \label{eq:Xpmf}
\begin{aligned}
  p^{(t)}_{S,I,R}(n_1,n_2,n_3)
  ={}&
  \binom{m_1}{m_1-n_1}(\exp(-h \beta m_2))^{n_1}(1-\exp(-h \beta
  m_2)^{m_1-n_1} \\
  & \binom{m_2}{n_3-m_3}(1-\exp(-h\gamma))^{n_3-m_3}\exp(-h\gamma)^{m_2-n_3+m_3}
\end{aligned}
\end{equation}
where $X_{t-1} = (m_1,m_2,m_3)$.

We can now define the transition matrix of the Markov chain. First, note that in
Equation~\eqref{eq:Xpmf} the dependency on $t$ can be changed to a
dependency on $\vec{m} = (m_1,m_2,m_3)$. Now, for all states $\vec{m},\vec{n} =
(n_1,n_2,n_3)$, we define:
\begin{equation}
\label{eq:transitionmatrix}
  P_{\vec{m} \vec{n}} = \begin{cases}
    p_{S,I,R}^{(\vec{m})}(n_1,n_2,n_3) & \text{if } n_1 \leq m_1 \text{ and }
    m_3 \leq n_3 \leq m_2 + m_3\\
    0 & \text{otherwise.}
  \end{cases}  
\end{equation}
Then, the matrix $\vec{P}$
is clearly stochastic. 

To conclude this section we state the following property which follows
immediately from the definition of $\vec{P}$ in \autoref{eq:transitionmatrix}.
\begin{lem}\label{lem:succ}
  Let $\vec{m} = (m_1,m_2,m_3),\vec{n}=(n_1,n_2,n_3)$ be states. Then, $0 < \vec{P}(\vec{m},\vec{n})$ if
  and only if $n_1 \leq m_1$ and $m_3 \leq n_3 \leq m_2 + m_3$.
\end{lem}

\section{Binomial chains}\label{sec:definitions}

Let $k \in \mathbb{N}$ such that $k \geq 1$. 
A \emph{binomial chain} (BC) is essentially a Markov chain $(S, s_0, P)$ such that $S \subseteq \mathbb{N}^k$ and whose transitions correspond to \emph{transfers} of individuals between \emph{compartments} modeled by the components of the state vectors. The probability of each individual transfer is based on a binomial distribution whose success probability is a function of the current state.

More formally, a BC $\mathcal{B}$ is a tuple $(\vec{v},\vec{T})$ where:
\begin{itemize}
    \item $\vec{v} \in \mathbb{N}^k$ is the initial state and
    \item $\vec{T} \in \mathbb{L}_k^{k \times k}$ is the \emph{transfer matrix}.
\end{itemize}
\noindent 
For intuition, $\vec{T}$ can be thought of as the adjacency matrix (if we care only about whether a function entry is the zero function or not) of a directed graph with $[k]$ as its vertices. Edges of said graph represent possible transfers of individuals from the source compartment of the edge to the target compartment.

\begin{exa}    
The transfer matrix $\vec{T}$ of the SIR binomial chain from \autoref{sec:sir} is given below.
\begin{equation}\label{eqn:t-sir}
    \vec{T} = \begin{pmatrix}
        0 & \vec{m} \mapsto h\beta m_2 & 0\\
        0 & 0 & \vec{m} \mapsto h\gamma\\
        0 & 0 & 0
    \end{pmatrix}
\end{equation}
It can be inferred from \autoref{fig:flow_diagram} and \autoref{eq:Xpmf}. \revone{R1.4}{(A formal description of how this is done follows, here we are just interested in conveying the intuition of the model.) Notice that the entries in the matrix correspond to the arguments of the exponentials in the expression that gives the transition probability.}\eoe
\end{exa}

Let us formalize the idea of making $\vec{T}$ into a (Boolean) adjacency matrix since the notation will be useful later. First, we extend the notion of matrix support to matrices with entries from $\mathbb{L}_k$. For $f \in \mathbb{L}_k$, we write $f \not\equiv 0$ to denote the fact that $f(\vec{x}) \neq 0$ for some $\vec{x} \in \mathbb{N}^k$. Then, the \emph{support} of a matrix $\vec{M} \in \mathbb{L}_k^{\ell \times m}$ is defined as follows.
\[
    \supp{\vec{M}} = \{(i,j) \in [\ell] \times [m] : M_{ij} \not \equiv 0 \}
\]
The adjacency matrix is $\restr{\vec{T}}{\mathbb{B}}$ where $(\restr{\vec{T}}{\mathbb{B}})_{ij} = 1$ if and only if $(i,j) \in \supp{\vec{T}}$.

\begin{exa}    
The adjacency matrix $\restr{\vec{T}}{\mathbb{B}}$ of the SIR binomial chain from \autoref{sec:sir} is given below.
\begin{equation}\label{eqn:tb-sir}
    \restr{\vec{T}}{\mathbb{B}} = \begin{pmatrix}
        0 & 1 & 0\\
        0 & 0 & 1\\
        0 & 0 & 0
    \end{pmatrix}
\end{equation}
This one can be obtained directly from \autoref{fig:flow_diagram}. \revone{R1.4}{Indeed, it is nothing more than the adjacency matrix of that graph showing what transfers ae possible between compartments.} \eoe
\end{exa}

The semantics of the BC $\mathcal{B}$ is given via the \emph{induced Markov chain} $\mathcal{C}_{\mathcal{B}} = (S,s_0,P)$ where $S = \mathbb{N}^{k}$ and $s_0 = \vec{v}$. For the transition probability function $P$, we have that $P(\vec{u},\vec{w}) > 0$ only if there exists $\vec{M} \in \mathbb{N}^{k \times k}$ such that:
\begin{equation}\label{eqn:constr-M}
\begin{aligned}
    \supp{\vec{M}} & {} \subseteq \supp{\vec{T}} & \text{(only allowed transfers)}\\
    \min_{i \in [k]} \max_{j \in [k]} M_{ij} & {} \leq u_i & \text{(valid binomial outcomes)}
\end{aligned}
\end{equation}
and $\vec{w}$ is the pointwise maximum of $\vec{0}$ and:
\[
    \vec{u} + \overbrace{(\vec{1}^\intercal \vec{M})^\intercal}^{\text{incoming transfers}} - \underbrace{ \vec{M}\vec{1}}_{\text{outgoing transfers}} 
\]
where $\vec{1}$ is a vector of all ones.
In other words, $\vec{w}$ is such that, for all $j \in [k]$ we have:
\begin{equation}\label{eqn:right-sum}
    w_j = \max\left(0, u_j + \sum_{i \in [k]} M_{ij} - \sum_{\ell \in [k]} M_{j\ell}\right).
\end{equation}

\begin{exa}\label{exa:sir-trans}
Let us turn once more to the SIR binomial chain from \autoref{sec:sir}. Recall that $k = 3$. Further consider states $\vec{u} = (10, 3, 2)^\intercal$ and $\vec{w} = (8,2,5)^\intercal$. We use matrices $\vec{M} \in \mathbb{N}^{k \times k}$ to 
encode information about transfers between compartments. To move from state $\vec{u}$ to $\vec{w}$ by transferring individuals along edges from \autoref{fig:flow_diagram}, thus nonzero entries from Equation \eqref{eqn:tb-sir}, we can use the following matrix.
\[
    \vec{M} = \begin{pmatrix}
        0 & 2 & 0\\
        0 & 0 & 3\\
        0 & 0 & 0
    \end{pmatrix}
\]
\revone{R1.4}{Moving $2$ susceptible individuals to the infectious compartment, and $3$ infectious ones to the recovered compartment, yields $\vec{w}$ if we start from $\vec{u}$. In symbols, this is just $\vec{w} = \vec{u} + (\vec{1}^\intercal\vec{M})^\intercal - \vec{M}\vec{1}$. In this example, it turns out that $\vec{M}$ is the unique matrix that satisfies that equation and the constraints imposed above. Furthermore, the pointwise maximum with $\vec{0}$ is unnecessary. 
However, uniqueness and positivity of $\vec{u} + (\vec{1}^\intercal\vec{M})^\intercal - \vec{M}\vec{1}$ are not guaranteed in general.}\eoe
\end{exa}

It is easy to see that, for any given $\vec{u}$, the set of $\vec{M}$ that
satisfy the constraints from \autoref{eqn:constr-M} is finite. We write
$\wit{\vec{u}}{\vec{w}}$ to denote the set of all matrices satisfying
\autoref{eqn:constr-M} and \autoref{eqn:right-sum}. Finally, the transition probability
$P(\vec{u},\vec{w})$ is defined as follows:
\begin{equation}\label{eqn:prob-formula}
    \sum_{\vec{M} \in \wit{\vec{u}}{\vec{w}}} 
    \prod_{(i,j) \in \supp{\vec{T}}} B(M_{ij}; u_i, 1 - \exp(-T_{ij}(\vec{u})))
\end{equation}
where $B(m;n,p)$ stands for the probability mass function of the binomial distribution:
\[
    B(m;n,p) = \binom{n}{m} p^m (1-p)^{n-m}
\]
with the convention that $0^0 = 1$.

\begin{exa}
    We continue with the situation from \autoref{exa:sir-trans}. \revone{R1.4}{Recall that we argued $\vec{M}$ was the unique matrix satisfying the imposed constraints. That is, $\{\vec{M}\} = \wit{\vec{u}}{\vec{v}}$.} Now, from \autoref{eq:Xpmf} we know that:
    \[
    \begin{aligned}
    P(\vec{u},\vec{w}) = {} &
    \binom{10}{2}(\exp(-h \beta 3))^{8}(1-\exp(-h \beta
  3)^{2} \\
  &{} \binom{3}{3}(1-\exp(-h\gamma))^{3}\exp(-h\gamma)^{0}
    \end{aligned}
    \]
    and this coincides with what we obtain from Equation \eqref{eqn:prob-formula} using Equation \eqref{eqn:t-sir} for the transfer matrix. Namely, we get:
    \[
    P(\vec{u},\vec{w}) = B(2; 10, 1-\exp(h\beta 3)) B(3; 3, 1-\exp(h\gamma))
    \]
    as expected.\eoe
\end{exa}

\begin{rem}[Negative populations]\label{rem:neg-pop}
    Note that, in general, the second condition from \autoref{eqn:constr-M} on matrices $\vec{M} \in \wit{\vec{u}}{\vec{w}}$ is necessary to avoid negative populations, but it is not sufficient.
    \revtwo{R2.2}{This is why the pointwise maximum with $\vec{0}$, from} \autoref{eqn:right-sum}, is needed.
    A natural stronger condition would be to ask that
    $\vec{M}\vec{1} \leq \vec{u}$. In order to keep the model as simple as possible and to avoid (further) complicating the transition probability expressions, the simpler (insufficient) condition is often preferred (see, e.g.,~\cite{ABRAMS2021100449}). In particular, using the stronger condition would require renormalizing \autoref{eqn:prob-formula} since it would exclude certain outcomes of the binomial distributions. 
\end{rem}

Having a nonempty set $\wit{\vec{u}}{\vec{w}}$ is not sufficient to guarantee a positive transition probability. This is because the success probability of some binomial distribution may still be $0$ for $\vec{u}$. Excluding that possibility gives us a sufficient and necessary condition.
\begin{lem}\label{lem:pos-prob}
    Let $\vec{u},\vec{w} \in \mathbb{N}^k$. We have $P(\vec{u},\vec{w}) > 0$ \revtwo{R2.3}{if and only if there exists $\vec{M} \in \wit{\vec{u}}{\vec{w}}$ such that $M_{ij} > 0$ implies $T_{ij}(\vec{u}) > 0$ for all $i,j \in [k]$.}
\end{lem}

\subsection{Interesting subclasses}

We now introduce some natural subclasses of binomial chains. The first two are already present in Kingma's work~\cite{Kingman_1969}. Let $(\vec{v},\vec{T})$ be a BC.
\begin{description}
    \item[Simple] We say that it is \emph{simple} if for all $i \in [k]$ and all $j \in [k]$, we have that $T_{ij}(\vec{u})$ can be written as $f(u_i)$.
    That is, all transfers from the $i$-th component depend only on the current number of individuals in that compartment.
    \item[Closed] We say that it is \emph{closed} if all of its transitions preserve the total population. In symbols, 
    for all $\vec{u},\vec{w} \in \mathbb{N}^k$ we have that:
    \[
    P(\vec{u},\vec{w}) > 0 \text{ implies } \norm{1}{\vec{u}} = \norm{1}{\vec{w}}.
    \]
    \item[Acyclic] We say it is \emph{acyclic} if $\restr{\vec{T}}{\mathbb{B}}$ is acyclic. This means that there exists no pair $(i,n) \in [k] \times \mathbb{N}$ such that $(\restr{\vec{T}}{\mathbb{B}})^n_{ii} > 0$.
\end{description}

Note that being simple and acyclic are properties that can be checked on the description of the BC (concretely, by inspecting the transfer matrix $\vec{T}$). In contrast, the definition of when a BC is closed seems to depend on the induced Markov chain. Fortunately, we have the following characterization of closed BCs.
\begin{lem}
    A BC $(\vec{v},\vec{T})$ is closed if and only if its transfer matrix $\vec{T}$ is such that every row has at most one nonzero entry. That is, 
    \(
        |\{(i',j) \in \supp{\vec{T}} : i = i'\}| \leq 1
    \)
    for all $i \in [k]$.
\end{lem}
\begin{proof}
If all rows of the transfer matrix have at most one nonzero entry, the stronger condition from \autoref{rem:neg-pop} coincides with the weaker one from our definition of BC. The contrapositive of the converse is also easy to establish: if the BC is not closed then it must be the case that a vector state with a negative component --- say, in dimension $i$ --- is reached. In turn, from the definition of the binomial distribution, this can only happen if the $i$-th row of $\vec{T}$ has at least two nonzero entries.
\end{proof}

In the rest of this work, we will mainly study acyclic and closed BCs as they cover interesting models used in epidemiology.

\subsection{Computational problems}
Suppose $k \in \mathbb{N},k \geq 1$ is fixed and that we are given a BC $\mathcal{B}$ as the tuple $(\vec{v},\vec{T})$ with functions from $\mathbb{L}_k$ represented as pairs of vectors with rational entries. In turn, assume integers are encoded in binary; and rationals, as pairs of integers representing the numerator and the denominator of the rational number (in reduced form).

The following problems are of practical interest in view of the applications of BCs.
\begin{description}
    \item[Termination] asks whether the BC almost surely reaches \emph{(final) states} $\vec{t}$ such that, in the induced Markov chain $\mathcal{C}_{\mathcal{B}}$, $P(\vec{t},\vec{t}) = 1$.
    \item[Time to termination] asks to compute the expected number of steps before termination, assuming the BC almost surely terminates.
\end{description}
Regarding the second problem, it is important to note that the value could be irrational. This is because of the exponential function used in the definition of the success probability of the binomial distributions in \autoref{eqn:prob-formula}.

In this work, we will primarily focus on algorithms to compute (rational approximations of) the expected time to termination.

\section{Acyclic binomial chains are absorbing}\label{sec:time2term-acyclic}
In this section, we recall the notion of absorbing Markov chain. Then, we state interesting properties of such chains that will be useful in the sequel. Finally, we argue that all acyclic BCs induce absorbing Markov chains.

\subsection{Absorbing Markov chains}
Let $\mathcal{C} = (S,s_0,P)$ be a Markov chain such that $S$ is finite.
We say that state $j \in S$ is \textit{reachable} from state $i \in S$ if and only if for some $t \in \mathbb{N}$:
\[
\Pr(X_t=j \mid X_0=i)>0.
\]
A state $i \in S$ is \textit{absorbing} if and only if $P_{ii}=1$. If the set of absorbing states $A \subseteq S$ of the Markov chain $\mathcal{C}$ is not empty and $A$ is reachable from all states, we say $\mathcal{C}$ is an \textit{absorbing Markov chain}. In an absorbing Markov chain, a state that is not absorbing is \emph{transient}. 

The transition matrix $\vec{P}$ of an absorbing Markov chain has special properties. Consider an ordering of $S$ such that the $k$ transient states are first, followed by the $\ell$ absorbing states. Now, the transition matrix will have the following \emph{canonical form}.
\begin{equation}\label{eqn:canon-form}
    \vec{P} = \begin{pmatrix}
        \vec{Q} & \vec{R}\\
        \vec{0} & \vec{I}
    \end{pmatrix}
  \end{equation}
Above, $\vec{R}$ is a nonzero $k \times \ell$ matrix; $\vec{Q}$, a $k \times k$ matrix; and $\vec{I}$ and $\vec{0}$, identity and zero matrices, respectively, of the appropriate dimensions.

\begin{prop}
    The matrix $\vec{I} - \vec{Q}$ has an inverse.
\end{prop}
\noindent 
The result above is well known. It is, for example, stated and proven in~\cite[Theorem 11.4]{grinstead}. The inverse of $\vec{I} - \vec{Q}$ is commonly written $\vec{N}$ and called the \emph{fundamental matrix}. The importance of this matrix will become clear in the next few paragraphs.

\subsection{Expected hitting times}
Let $A \subseteq S$ be a set of \emph{target states}.
We write $\tau_A$ to denote the first \textit{hitting time} of a state in $A$.
Note that $\tau_A$ can take countably many values only and they all are nonnegative. Hence, its expectation, denoted $\Ex[\tau_A]$, satisfies the following (see also~\cite[Section 1.3]{norris98}).
\begin{equation}\label{eqn:ex-hitting}
    \Ex[\tau_A] = \sum_{t = 0}^\infty t \Pr(\tau_A = t) + \infty \Pr(\tau_A = \infty)
\end{equation}

Let $i \in S$ be a state and write $k^A_i$ for the value $\Ex[\tau_A]$ where the initial state $s_0$ of $\mathcal{C}$ is replaced by $i$. The following characterization of the expected hitting times will be useful later. The result is well known and can be found, for instance, in \cite[Theorem 1.3.2]{norris98}.
\begin{prop}\label{thm:norris-hit}
  The vector of expected hitting times $\vec{k^A} = (k^A_i : i \in S)$ is the minimal (w.r.t. the product order) nonnegative solution to the following system.
    \begin{equation}
        \begin{cases}
            k^A_i = 0 & \text{if } i \in A\\
            k^A_i = 1 + \sum_{j \not\in A} P_{ij} k^A_j & \text{otherwise.}
        \end{cases}
    \end{equation}
\end{prop}

\subsection{Expected hitting times in absorbing Markov chains}\label{sec:ex-hit-abs}
For absorbing Markov chains, the probability that the process reaches an
absorbing state is one. This, in turn, means that the expected hitting times for the set $A \subseteq S$ of absorbing states are always finite.
In fact, a formula for the vector of expected hitting times exists in terms of the fundamental matrix (see, e.g.~\cite[Theorem 11.5]{grinstead}).
\begin{prop}\label{thm:abs-hit}
    Consider an absorbing Markov chain with absorbing set of states $A \subseteq S$ and fundamental matrix $\vec{N}$ and let $\vec{k^A} = (k_i^A : i \in S)$ be the vector of expected hitting times. Then, 
    \(
        (k^A_i : i \in S \setminus A) = \vec{N} \vec{1}
    \)
    and
    \(
        (k^A_i : i \in A) = \vec{0}.
    \)
\end{prop}

\subsection{Acyclic binomial chains are absorbing}\label{sec:acyclic-bcs}
We will now argue that every acyclic binomial chain $(\vec{T},\vec{v})$ induces an absorbing Markov chain. The summary of our approach is as follows. First, we will define a total order on state vectors. Then, we will establish that the transitions of the induced Markov chain respect this order. Finally, we will also prove that acyclic binomial chains induce finite Markov chains. The claim will follow directly from these properties.

\begin{thm}\label{thm:acyclic-absorbing}
    Let $\mathcal{B}$ be an acyclic BC. Then, its induced Markov chain $\mathcal{C}_{\mathcal{B}}$ is absorbing.
\end{thm}
\noindent 
Since $\restr{\vec{T}}{\mathbb{B}}$ is acyclic, we can use it to sort $[k]$ topologically. That is to say, we can assume that for all $i , j \in [k]$ the following holds.
\begin{equation}\label{eqn:lex-idx}
    T_{ij} \not\equiv 0 \implies i < j
\end{equation}
This means that any $\vec{M} \in \wit{\vec{u}}{\vec{w}}$ will be upper triangular with zeros in the diagonal.
Based on this observation, we focus on the lexicographic order on vector states: we write $\vec{w} \lexleq \vec{u}$ if and only if $\vec{w} = \vec{u}$ or there exists $j \in [k]$ such that $w_j < u_j$ and $w_i = u_i$ for all $1 \leq i < j$. We claim that the transitions of the induced Markov chain respect this order.
\begin{lem}\label{lem:local-order}
    Let $\vec{u}$ and $\vec{w}$ be states of the acyclic BC and $P$ the transition probability function of its induced Markov chain. If $0 < P(\vec{u},\vec{w})$ then $\vec{w} \lexleq \vec{u}$ and $\norm{1}{\vec{w}} \leq \norm{1}{\vec{u}}k$.
\end{lem}
\begin{proof}
    First, assume $\vec{u} \neq \vec{w}$ as otherwise the claim holds trivially.
    Let $j' \in [k]$ be the index of the first row of $\vec{M}$ containing some nonzero entry. Hence, we have the following.
    \[
        \forall j < j' : \sum_{\ell \in [k]} M_{j\ell} = 0
    \]
    From the triangularity observation above we also get that all columns of $\vec{M}$ with index $j \leq j'$ have only zeros. Hence, we also have the following.
    \[
       \forall j \leq j' : \sum_{i \in [k]} M_{ij} = 0 
    \]
    It follows from \autoref{eqn:right-sum} that $w_j = u_j$ for all $j < j'$ and $w_{j'} = u_{j'} - \sum_{\ell \in [k]} M_{j'\ell}$. Since the entries of $\vec{M}$ are nonnegative and its $j'$-th row has some nonzero entry, we get $w_{j'} < u_{j'}$ and thus $\vec{w} \lexleq \vec{u}$.

    For the second property we also rely on $\vec{M}$ being upper triangular. Because of that fact, we can rewrite \autoref{eqn:right-sum} as follows for all $j \in [k]$.
    \begin{align*}
        w_j & {}= \max\left(0, u_j + \sum_{i < j} M_{ij} - \sum_{\ell > j} M_{j\ell}\right)\\
        & {} \leq \max\left(0, u_j + \sum_{i < j} M_{ij}\right) & \vec{M} \text{ is nonneg.}\\
        & {} = u_j + \sum_{i < j} M_{ij} & \vec{u},\vec{M} \text{ are nonneg.}\\
        & {} \leq \sum_{i \leq j} u_i & \text{by \autoref{eqn:constr-M}}
    \end{align*}
    It thus follows that $\norm{1}{\vec{w}} \leq \norm{1}{\vec{u}}k$ as claimed.
\end{proof}

To conclude the proof of \autoref{thm:acyclic-absorbing} we also argue that acyclic BCs are finite. This follows from \autoref{lem:local-order} and Dickson's lemma since the lexicographic order is a linear extension of the product order. However, it will be useful later to have an explicit bound on the size of the set of states of the BC. Our bound will rely on \autoref{lem:local-order} and the fact that every state that is reachable in the (Markov chain induced by the) BC is reachable in a small number of steps. As a stepping stone, we prove the claim for the case when the support of intermediate states stays the same.
\begin{lem}\label{lem:same-supp}
    Let $\vec{u^{(1)}}, \dots, \vec{u^{(n)}} \in \mathbb{N}^k$ be vectors with the same support and $t \in \mathbb{N}$ such that $t \geq 1$. If $\Pr(X_{t+n} = \vec{u^{(n)}}, \dots, X_t = \vec{u^{(1)}}) > 0$ then there exists $\ell < k$ such that $\Pr(X_{t+\ell} = \vec{u^{(n)}}, X_t = \vec{u^{(1)}}) > 0$.
\end{lem}
\begin{proof}
    If $n < k$, the claim holds trivially. Hence, we focus on the case where $n \geq k$.
    
    We will need some additional notation. From the assumptions in the claim, there are matrices $\vec{M^{(1)}}, \dots, \vec{M^{(n-1)}}$ such that $\vec{M^{(\ell)}} \in \wit{\vec{u^{(\ell)}}}{\vec{u^{(\ell+1)}}}$ for all $1 \leq \ell < n$. In addition, the matrices also satisfy the following for all $1 \leq \ell < n$.
    \[
    \prod_{(i,j) \in \supp{\vec{T}}} B(M^{(\ell)}_{ij}; u^{(\ell)}_i, 1 - \exp(-T_{ij}(\vec{u^{(\ell)}}))) > 0
    \]
    \revtwo{R2.6}{Finally, we write $\vec{M_i^{(\ell)}}$ to denote the $i$-th row $(M^{(\ell)}_{i1},M^{(\ell)}_{i2},\dots,M^{(\ell)}_{ik})$ of $\vec{M^{(\ell)}}$.}

    \revtwo{R2.6}{We will now define a second sequence of states $\vec{w^{(1)}} \dots \vec{w^{(k)}}$, such that $\vec{w^{(1)}} = \vec{u^{(1)}}$ and $\vec{w^{(k)}} = \vec{u^{(n)}}$, with their corresponding matrices $\vec{N^{(1)}}, \dots, \vec{N^{(k-1)}}$.} It will be clear that $\vec{N^{(\ell)}} \in \wit{\vec{w^{(\ell)}}}{\vec{w^{(\ell+1)}}}$, for all $1 \leq \ell < k$, based on the triangularity of the $\vec{M^{(\ell)}}$ and how we will define our $\vec{N^{(\ell)}}$. Afterwards, we will argue that the corresponding products of the binomial probability mass functions are positive. \revtwo{R2.6}{The intuition is that $\vec{N_i^{(\ell)}}$ groups all transfers encoded by the $\vec{M^{(\ell')}}$, for $1 \leq \ell' < n$, from the $i$-th compartment to other compartments $j \geq i$.}

    We first construct the matrices. For $1 \leq \ell < k$ and all $i \in [k]$, we set:
    \[
        \vec{N^{(\ell)}_i} = \begin{cases}
            \sum_{1 \leq \ell' \leq n} \vec{M^{(\ell')}_{i}} & \text{if } i = \ell\\
            \vec{0} & \text{otherwise.}
        \end{cases}
    \]
    Now, we define $\vec{w^{(\ell+1)}}$ as the pointwise maximum of $\vec{0}$ and $\vec{w^{(\ell)}} + (\vec{1}^\intercal \vec{N^{(\ell)}})^\intercal - \vec{N^{(\ell)}}\vec{1}$. 
    
    Let $\ell \in [k]$ be arbitrary. By triangularity of the original witness matrices $\vec{M^{(\ell')}}$ and the definition of the new ones $\vec{N^{(\ell')}}$, we get that $w^{(\ell)}_i = u^{(n)}_i$ for all $i < \ell$ and $w^{(\ell)}_i \geq u^{(1)}_i$ for all $i \geq \ell$. In particular, since all $\vec{u^{(\ell')}}$ have the same support, this means that: 
    \[
    \supp{\vec{w^{(\ell)}}} =
    \supp{\vec{u^{(n)}}} =
    \supp{\vec{u^{(1)}}} =
    \supp{\vec{u^{(\ell)}}}.
    \]
    Recall that every $T_{ij}$ can be written as $\vec{x} \mapsto \vec{a}^\intercal \vec{x} + b$ with $\vec{a}$ and $b$ being nonnegative. Hence, if $T_{ij}(\vec{x}) > 0$ and $\supp{\vec{y}} = \supp{\vec{x}}$, for vectors $\vec{x},\vec{y} \in \mathbb{N}^k$, then $T_{ij}(\vec{y}) > 0$. We thus conclude that for all $1 \leq \ell' < \ell$ the following holds.
    \[
    \prod_{(i,j) \in \supp{\vec{T}}} B(N^{(\ell')}_{ij}; w^{(\ell')}_i, 1 - \exp(-T_{ij}(\vec{w^{(\ell')}}))) > 0
    \]
    Therefore, as required, $\Pr(X_{t+k-1} = \vec{u^{(n)}}, X_t = \vec{u^{(1)}}) > 0$.
\end{proof}

We can now state a concrete bound on the norm of reachable states in an acyclic BC.
\begin{lem}\label{lem:acyclic-bc-finite}
    Let $\mathcal{B} = (\vec{v},\vec{T})$ be an acyclic BC. The set $R$ of reachable states in its induced Markov chain is finite. Moreover, for all $\vec{w} \in R$ we have that $\norm{1}{\vec{w}} \leq \norm{1}{\vec{v}} \exp((\ln k) 2^{k^2})$.
\end{lem}
\begin{proof}
    First, we prove that \autoref{lem:same-supp} can be generalized to when the support of intermediate states does change. This will cost us an exponential in terms of $k$. Concretely, we claim that any reachable state $\vec{w} \in \mathbb{N}^k$ can be reached in at most $2^{k^2}$ steps. Towards a contradiction, suppose the shortest run witnessing that $\vec{w}$ is reachable (from $\vec{v}$) is longer. Then, by pigeonhole principle, the run contains an infix of length at least $k$ such that all states in it have the same support. But then, by \autoref{lem:same-supp}, we can shorten that infix and obtain a shorter witness.

    Second, to obtain the bound on the norm of reachable states $\vec{w}$ we just need to compose $2^{k^2}$ times the bound from \autoref{lem:local-order}. We thus obtain:
    \[
        \norm{1}{\vec{w}} \leq \norm{1}{\vec{v}} k ^{2^{k^2}} = \norm{1}{\vec{v}} \exp((\ln k) 2^{k^2})
    \]
    as claimed.
\end{proof}

We close this section with the observation that the bound from the previous claim is (asymptotically) good. Indeed, since $k$ is assumed to be fixed, the bound says that the $1$-norm of all reachable states is linear in that of the initial state, i.e. $O(\norm{1}{\vec{v}})$.

\section{Approximating the time to termination in polynomial space}\label{sec:exhit-approx}
In this section, we study the computational complexity of the (approximate)
time-to-termination problem for acyclic binomial chains. 
Somewhat surprisingly, despite the induced Markov chain $\mathcal{C}_{\mathcal{B}}$ of a given BC $\mathcal{B}$ being exponentially larger than the size of the encoding of $\mathcal{B}$ (see \autoref{lem:acyclic-bc-finite} and recall the components of $\vec{v}$ are encoded in binary), we can compute the expected time to termination using polynomial space only. 
Intuitively, our algorithm consists of a composition of three polynomial-space transducers that have to deal with numbers whose binary encoding may be large (read, exponential). One transducer reads bits of the encoding of the success probabilities from the binomial distributions, another reads bits of the encoding of the transition probabilities, and a last one computes bits of the expected time to termination. However, before we
get to that point, we will have to address a small yet important issue: 
While the input to our computational problems is finite,
the transition probabilities of the induced Markov chain may be irrational.

\subsection{Rational approximations of induced probabilities}
We follow a very simple approach. In short, we use the Taylor expansion
of $e^{-x}$ to get a rational approximation thereof. We also derive a bound on
the error of our approximation using standard calculus.

Let $x \in \mathbb{Q}_{\geq 0}$. We assume $x$ is given as a pair of positive integers
encoding the numerator and denominator, i.e. $x = \frac{a}{b}$ for $a,b \in
\mathbb{N}$ with $a,b > 0$. We study the function $f(x) = e^{-x}$. The $k$-th
order Taylor approximation of $f$ is the following polynomial.
\[
  P_k(x) = \sum_{i=0}^k \frac{(-x)^i}{i!}
\]
Now, since $-1 \leq f^{(k)}(x) \leq 1$ for all $k \in \mathbb{N}$ and all $x
\in \mathbb{Q}_{\geq 0}$, the remainder
\(
  R_k(x) = f(x) - P_k(x)
\)
can be bounded as below using Taylor's inequality.
\begin{equation}\label{eqn:bnd-remainder}
  |R_k(x)| \leq \frac{x^{k+1}}{(k+1)!}
\end{equation}

We can now state a sufficient bound on how large $k$ should be to have a small
error in our approximation.
\begin{lem}\label{lem:taylor-error}
  Let $r \in \mathbb{N}$. If $k \geq 2a^2 + r$ then $|R_k(x)| \leq 2^{-r}$.
\end{lem}
\begin{proof}
  We start with a simple observation: There exists $K \in \mathbb{N}$ such
  that all $k \geq K$ satisfy $x,x^2 < k$. One can, for instance, take $K =
  a^2$. Now, we claim that for all $k \geq 2K$ the following hold.
  \[
    \frac{x^k}{k!} < \left(\frac{x}{K} \right)^{k - 2K} < 1
  \]
  The right inequality holds because of our choice of $K$. For the left
  inequality, we observe the following.
  \begin{align*}
    \frac{x^k}{k!} & {} < \left(\frac{x}{K} \right)^{k - 2K}
    \frac{x^{2K}}{(2K!)} & \text{because } K^{k-2K} < \frac{k!}{(k-2K)!}\\
    & {} < \left(\frac{x}{K} \right)^{k - 2K}
    \frac{(x^{2})^K}{K^K} & K^K < \frac{(2K)!}{K!} < (2K)!\\
    & {} < \left(\frac{x}{K} \right)^{k - 2K} &
    \frac{x^2}{K} < 1
  \end{align*}
  This means that for $k \geq 2K$ the error decreases exponentially.

  It still remains to determine how much
  larger than $2K$ do we need $k$ to be so that we get an error of at most
  $2^{-r}$. For this, consider the following.
  \begin{align*}
    \left(\frac{x}{k}\right)^i \leq 2^{-r} & {} \iff i \lg \frac{x}{K} \leq -r \\
    & {} \iff i \geq \frac{-r}{\lg \frac{x}{K}} & \text{since } x < K\\
    & {} \iff i \geq \frac{r}{\lg K - \lg x} = \frac{r}{\lg a^2 - \lg
      \frac{a}{b}}\\
    & {} \iff i \geq \frac{r}{\lg a + \lg b}
  \end{align*}
  Since the last inequality holds when $i \geq r$, we can choose $k \geq 2a^2 + r$. The result thus follows by \autoref{eqn:bnd-remainder}.
\end{proof}

From the preceding discussion we get a rational approximation of the success
probability of the underlying binomial distributions in our model. We still
need to determine how to compute it efficiently. It will turn out that we can
do so using only polynomial space (in terms of the number of bits required to
write $a$, $b$, and $r$ in binary). To be precise, we will be able to query
the $i$-th bit of the numerator or denominator (encoded in binary) of the
rational approximation.

\begin{lem}\label{lem:pspace-succ-prob}
    Given $a, b, r, i \in \mathbb{N}$ such that $a,b,r,i > 0$, all encoded in
    binary, we can compute the $i$-th bit of $n$ or $d$ in
    $\nicefrac{n}{d} = P_k(\nicefrac{a}{b})$, where $k = 2a^2 +
    r$, using only space $(\lg (abri))^{O(1)}$.
\end{lem}

Before we go into the proof of the claim, some definitions are in order.
We write $\mathbf{NC}^i$ for the class of decision problems solvable in time
$O(\log^i n)$, with $n$ the size of the input, on a parallel computer with a
polynomial number of processors~\cite{vollmer99,arorabarak09}.
It is known that any problem in $\mathbf{NC}^i$ can be solved deterministically using
space $O(\lg^i(n))$~\cite[Theorem 4]{DBLP:journals/siamcomp/Borodin77}. As a concrete example with $i = 1$, one can show that computing (a chosen bit of) the sum
or product of a list of binary-encoded integers is in
$\mathbf{NC}^1$~\cite[Chapter 1]{vollmer99} by exploiting associativity to realize the operations via binary splitting. For an example with $i = 2$, we ask, given a list of lists,
to compute the sum of the products (of the inner lists). Based on the previous example, this can be done in
$\mathbf{NC}^2$. 

\begin{proof}[Proof of \autoref{lem:pspace-succ-prob}]
  Recall the form of the approximation we are proposing.
  \[
    P_k(\nicefrac{a}{b}) = \frac{n}{d} = \frac{1}{b^k k!} \sum_{i=0}^k
    (-a)^i b^{k-i} (k-i)!
  \]
  This means that $d$ can be taken to be some power of $b$ and the factorial
  of $k$ while $n$ is a sum of products of powers factorials. From a
  complexity point of view, the complication is that $k \geq a^2$ and $a$ is given in binary, so $n$ and $d$ could require exponentially many bits
  to represent. (This is also why we focus on determining the value of a
  single bit only.) Note that computing $n$ and $d$ both amount to computing (sums of)
  products of at most $k^3$ binary-encoded integers. Recall that this problem
  is in $\mathbf{NC}^2$ when $k$ is small, i.e. given in unary.  This means it can also be solved using polylogarithmic space, so we conclude, due to $k$ being exponentially large, that we can solve it using polynomial space only.
\end{proof}
\noindent 
The argument used above will be repeated two times in the sequel to establish
that computing bits of transition probabilities in the induced Markov chain
and bits of entries of its fundamental matrix can be done using polynomial
space only. Finally, we observe that we can also query bits of an
approximation of $1-\exp(-x)$ in polynomial space by following all the same
steps while changing $P_k(x)$ to remove the first summand and flip all signs.

\subsection{Bits of the induced probabilities in polynomial space}
Consider a given acyclic binomial chain $(\vec{v},\vec{T})$. The main message in this subsection is that we can query bits of the numerator and denominator of \autoref{eqn:prob-formula} (reproduce explicitly below, for convenience) in polynomial space with respect to the encoding size of the given BC. Furthermore, we establish bounds regarding the accumulation of the error because of our usage of a rational approximation of the success probabilities (cf. previous subsection).
\begin{equation}\label{eqn:prob-form1}
    \sum_{\vec{M} \in \wit{\vec{u}}{\vec{w}}} 
    \prod_{(i,j) \in \supp{\vec{T}}} \binom{u_i}{M_{ij}} (\overbrace{1 - \exp(-T_{ij}(\vec{u}))}^{\text{computable in \textbf{PSPACE}}})^{M_{ij}} (\underbrace{\exp(-T_{ij}(\vec{u}))}_{\text{this too}})^{u_i - M_{ij}}
\end{equation}
Recall that all integers and rational numbers are encoded in binary. Now, since $\vec{T}$ is given as an explicit matrix, the product over $(i,j) \in \supp{\vec{T}}$ is small (i.e. a product of a polynomial number of terms w.r.t. the size of the input). The sum over $\vec{M} \in \wit{\vec{u}}{\vec{w}}$ is exponential though, because of the binary encoding. Moreover, the factorials arising from the binomial coefficients are also large since the components of $\vec{u}$ are given in binary.

We want to argue, as in the previous subsection, that sums of products being in $\mathbf{NC}^2$ and thus also in space $O(\lg^2(n))$, with $n$ the size of the input, gives us our polynomial-space result. We first need to take care that the numerator and denominator are clearly sums of products or just products. For this, we rewrite \autoref{eqn:prob-form1} as follows, where we are already using the approximation $P_k$ from the previous section instead of the (possibly irrational) success probabilities. We write $\tilde{P}$ for this approximate transition probability function. 
\begin{align}
    \tilde{P}(\vec{u},\vec{w}) = {} & \sum_{\vec{M} \in \wit{\vec{u}}{\vec{w}}} 
    \prod_{(i,j) \in \supp{\vec{T}}} \binom{u_i}{M_{ij}} \left(\frac{a_{ij}(\vec{u})}{b_{ij}(\vec{u})}\right)^{M_{ij}} \left(\frac{c_{ij}(\vec{u})}{d_{ij}(\vec{u})}\right)^{u_i - M_{ij}} \label{eqn:sop-with-err}\\
    {}={} & \mu \sum_{\vec{M}} \label{eqn:num}
    \prod_{(i,j)} \binom{u_i}{M_{ij}} a_{ij}(\vec{u})^{M_{ij}}  c_{ij}(\vec{u})^{u_i - M_{ij}} \prod_{\vec{M'} \in \wit{\vec{u}}{\vec{v}} \setminus \{\vec{M}\}} b_{ij}(\vec{u})^{M'_{ij}} c_{ij}(\vec{u})^{u_i - M'_{ij}}\\
    & \text{where } \mu \text{ is } \prod_{\vec{M} \in \wit{\vec{u}}{\vec{w}}} \prod_{(i,j) \in \supp{\vec{T}}} b_{ij}(\vec{u})^{-M_{ij}} d_{ij}(\vec{u})^{M_{ij} - u_i} \label{eqn:denom}
\end{align}
Above, it is clear that \autoref{eqn:num} can be used as the numerator of the result and \autoref{eqn:denom} as the reciprocal of the denominator. These are both in the form of a sum of products or a product, as required.

\begin{lem}\label{lem:trans-prob-approx}
    Given a BC $(\vec{v},\vec{T})$, state vectors $\vec{u},\vec{w} \in \mathbb{N}^k$, and $r,i \in \mathbb{N}$ (in binary) with $r,i > 0$, we can compute the $i$-th bit of $n$ or $d$ in $\nicefrac{n}{d} = \tilde{P}(\vec{u},\vec{w})$ using only polynomial space. Moreover, $|P(\vec{u},\vec{w}) - \tilde{P}(\vec{u},\vec{w})| \leq 2^{-r}$.
\end{lem}
\begin{proof}
    The argument to prove $\tilde{P}(\vec{u},\vec{w})$ can be computed using
    polynomial space is the same as the one used in the proof of
    \autoref{lem:pspace-succ-prob}. We just need to make use of
    \autoref{eqn:num}, \autoref{eqn:denom}, and \autoref{lem:pspace-succ-prob}
    to obtain (exponentially long) sums of products of terms whose bits can be
    queried in polynomial space. It remains to argue that the error bound
    holds, and that without having to make $r$ (and thus $k$) in
    our use of \autoref{lem:pspace-succ-prob} too large.

    Consider \autoref{eqn:sop-with-err} and note that the only terms with
    errors are the fractions arising from our approximation of the success
    probabilities. Multiplying numbers $x,y$ with $0 \leq x, y < 1$
    approximated with an error $0 < \varepsilon < 1$ results in at most
    tripling the error. Similarly, addition results in at most doubling the
    error (this holds in general though, even if the assumptions stated for
    multiplication do not hold). Now, by analyzing the exponentiation using
    repeated squaring, we have at most polynomial tripling of the error in the
    terms $a_{ij}(\vec{u})^{M_{ij}} b_{ij}(\vec{u})^{-M_{ij}}$ and
    $c_{ij}(\vec{u})^{u_i - M_{ij}} d_{ij}(\vec{u})^{M_{ij} - u_i}$. Now,
    while the product of these two terms is smaller than $1$, the binomial
    coefficient is not. In this case, multiplication affects the error much
    more and we get that it is amplified by at most
    $\norm{\infty}{\vec{u}}^{\norm{\infty}{\vec{u}}}$. Fortunately, the
    product with the binomial coefficient once more yields a value smaller
    than $1$. So, to summarize these observations in symbols, if we started
    with an error of at most $2^{-r'}$ from our use of
    \autoref{lem:taylor-error} and \autoref{lem:pspace-succ-prob} then we get
    that:
    \[
        |P(\vec{u},\vec{w}) - \tilde{P}(\vec{u},\vec{w})| \leq \norm{\infty}{\vec{u}}^{\norm{\infty}{\vec{u}}}3^c 2^{-r'} = 2^{(\lg \norm{\infty}{\vec{u}})\norm{\infty}{\vec{u}}}3^c 2^{-r'} \leq 2^{\norm{\infty}{\vec{u}}^2}3^c 2^{-r'}
    \]
    for some small $c \in \mathbb{N}$ --- which we could even assume to be encoded in unary. Since $3^c 4^{-c} \leq 1$, it suffices to use $r' = 2rc \norm{\infty}{\vec{u}}^2$ when appealing to \autoref{lem:pspace-succ-prob} to get the required bounds. Importantly, this means the complexity bound claimed above does hold since $r'$ requires only linearly many more bits to be encoded in binary compared to $r$ and $\vec{u}$.    
\end{proof}

The last step of our algorithm is arguably the most complex. We intend to
compute the fundamental matrix, i.e. the inverse of the induced Markov chain
$\mathcal{C}_{\mathcal{B}}$ using polynomial space only. For that, and also as
a sanity check, we first wonder how the absolute error $2^{-r}$ can be chosen
to be certain it is smaller than the transition probabilities.
\begin{rem}[Avoiding small probabilities with large
  errors]\label{rem:small-rel-error}
  Let $\vec{u}$ and $\vec{w}$ be vector states reachable from $\vec{v}$ such that $\tilde{P}(\vec{u},\vec{w}) > 0$ and consider the question of determining a sufficient lower bound for $r$ so
  that $\tilde{P}(\vec{u},\vec{w}) > 2^{-r}$.
  Ideally $r$ can be encoded in binary using a polynomial number of bits with
  respect to the rest of the input. To see that this is indeed the case, we
  can study \autoref{eqn:prob-form1} and find an upper bound for it under the
  assumption that it is not zero.
  \begin{equation}\label{eqn:bnd-probs}
  \begin{aligned}
    & \sum_{\vec{M} \in \wit{\vec{u}}{\vec{w}}} 
      \prod_{(i,j) \in \supp{\vec{T}}} \binom{u_i}{M_{ij}} (1 -
      \exp(-T_{ij}(\vec{u})))^{M_{ij}} (\exp(-T_{ij}(\vec{u})))^{u_i -
      M_{ij}}\\
    {} \geq {} &
      \min_{\vec{M} \in \wit{\vec{u}}{\vec{w}}} 
      \min_{(i,j) \in \supp{\vec{T}}} \left((1 -
      \exp(-T_{ij}(\vec{u})))^{M_{ij}} (\exp(-T_{ij}(\vec{u})))^{u_i -
    M_{ij}}\right)^{k^2}\\
    {} \geq {} &
      \min_{(i,j) \in \supp{\vec{T}}}
      \min\left((1 -
      \exp(-T_{ij}(\vec{u})))^{u_i}
    (\exp(-T_{ij}(\vec{u})))^{u_i}\right)^{k^2}
  \end{aligned}
  \end{equation}
  Now, write $n$ for the number of bits used to encode the BC
  $(\vec{v},\vec{T})$. By
  \autoref{lem:acyclic-bc-finite}, the number of bits required to encode
  a state $\vec{u}$ reachable from $\vec{v}$ is at most $n + k2^{k^2}$.
  Furthermore, the number of bits required to encode the vectors for the
  functions $T_{ij}$ is also $n$ so $T_{ij}(\vec{u})$, for some reachable
  $\vec{u}$, satisfies the following when it is nonzero:
  \[
    \frac{1}{2^{n+1}} \leq T_{ij}(\vec{u}) \leq 2^{(k+2)n + k2^{k^2}}.
  \]

  From the above inequalities we get that $\exp(-T_{ij}(\vec{u}))$ can be made
  closer to $0$ than to $1$ (assuming $T_{ij}(\vec{u})$ is not $0$). Together
  with \autoref{eqn:bnd-probs} we get that if $\tilde{P}(\vec{u},\vec{w})$ is
  not zero then it satisfies the following.
  \begin{equation}\label{eqn:lower-bnd-prob}
    \tilde{P}(\vec{u},\vec{w}) > \exp(-2^{(k+2)n + k2^{k^2}})^{k^2} \geq 2^{-2k^2 2^{(k+2)n + k2^{k^2}}}
  \end{equation}
  Hence, by choosing $r$ larger than $2k^2 2^{(k+2)n + k2^{k^2}}$, which can
  be encoded in a polynomial number of bits in $n$ (because $k$ is fixed), we
  get the desired inequality.
\end{rem}

\subsection{Time to termination in polynomial space} Given, a
nonsingular $n \times n$ matrix $\vec{A}$ with $n$-bit integer entries
(encoded in binary), we must output its inverse $\vec{A}^{-1}$ in the form of
a pair $(\adj(\vec{A}),\det(\vec{A}))$ consisting of the adjugate and the
determinant of $\vec{A}$. It is known that computing chosen bits of the
numerator or denominator of a chosen entry of $A^{-1}_{ij} =
\nicefrac{\adj(\vec{A})_{ij}}{\det(\vec{A})}$ is in
$\mathbf{NC}^2$~\cite[Proposition 5.2]{cook85}.

Recall that the set of reachable states in the Markov chain induced by a given
BC $(\vec{v},\vec{T})$ is finite, yet exponential (due to the binary encoding
of the integers)~\autoref{lem:acyclic-bc-finite}.
Now, we would like to appeal to \autoref{thm:abs-hit} in combination with the
above observations to obtain a polynomial-space algorithm for obtaining the
bits of the numerator and denominator of a chosen $k^A_i$ from the induced
Markov chain. 

\begin{thm}\label{thm:eoe-pspace}
  Given a BC $(\vec{T},\vec{v})$ and $r, i \in \mathbb{N}$ in binary with $r, i >
  0$, we can compute the $i$-th bit of $n$ or $d$ in $\nicefrac{n}{d} =
  \tilde{k}^A_{\vec{v}}$ using only polynomial space. Moreover,
  $|\tilde{k}^A_{\vec{v}} - k^A_{\vec{v}}| \leq 2^{-r}$.
\end{thm}
\begin{proof}
We first reorder the states to get $\vec{P}$ in its
canonical form (see \autoref{eqn:canon-form}). We
rely on the lexicographic ordering introduced in \autoref{sec:acyclic-bcs} to achieve this. Concretely, if we want the $i$-th state $\lexleq$-smaller than
$\vec{v}$, we can enumerate all vector states satisfying the bound from
\autoref{lem:acyclic-bc-finite} to find the maximal
state $\vec{w}$ smaller than $\vec{v}$ and repeat $i-1$ times from $\vec{w}$, all using polynomial space only. 
Further note that checking whether a state $\vec{w}$ is absorbing can be
implemented in polynomial space using \autoref{lem:trans-prob-approx} by
enumerating all states not equal to $\vec{w}$ and confirming the probability
to transition to them is $0$. (This last check can even be done
without having to approximate the success probabilities!)

Before applying the complexity result for inverting a matrix, we need to deal with the fact that we want to invert a matrix of rational numbers and not one of integers. The natural approach would be to factor out a common denominator (as we have done in previous subsections). However, because of the shape of $\vec{I} - \vec{Q}$ in our case, this is not necessary. By our choice of reordering of the states and
\autoref{lem:local-order}, we have that $\vec{I} - \vec{Q}$ is upper triangular. In turn, this means that its inverse $\vec{N} = \nicefrac{\adj(\vec{I} - \vec{Q})}{\det(\vec{I} - \vec{Q})}$ and thus also $\adj(\vec{I} - \vec{Q})$ are upper triangular. Since the nonzero entries of the the adjugate $\adj(\vec{I} - \vec{Q})$ are obtained as signed determinants of minors of $\vec{I} - \vec{Q}$ obtained by removing rows $\vec{u}$ and columns $\vec{w}$ with $ \lexleq \vec{w}$. Such minors will necessarily be upper triangular too. Therefore, their determinants are just products of entries of $\vec{I} - \vec{Q}$. Importantly, no sum is needed, so we can focus independently on numerators or denominators of the entries of $\vec{I} - \vec{Q}$ based on whether we want bits of the numerator or denominator of an entry of $\vec{N}$.

From the discussion above and the $\mathbf{NC}^2$ bound for matrix inversion
(and the polynomial-space algorithms from previous subsections) that the bits
of fundamental matrix $\vec{N}$ can be queried using polynomial space, we
still need to establish the same for $\vec{N}\vec{1}$. This is not a problem
since multiplying by $\vec{1}$ on
the right amounts to adding rows of $\vec{N}$, which can be done in
$\mathbf{NC}^1$ if $\vec{N}$ is given explicitly. Altogether, inverting a matrix and then
adding its rows can be done in $\mathbf{NC}^3$ thus also sequentially while
using only $O(\lg^3(n))$ space, with $n$ the size of the input. This means that
in our (exponentially large) induced Markov chain, the bits of the (numerator
and denominator) of $\vec{N}\vec{1}$ can be queried in polynomial space. When using the approximation of the success probabilities, we write
$\vec{\tilde{k}^A}$ instead of $\vec{k^A}$.

We have already argued why the complexity bound in \autoref{thm:eoe-pspace}
holds. It remains for us to prove that the error bounds are true. First, observe that we can ensure the error of the entries of $\adj(\vec{I}-\vec{Q})$ and $\det(\vec{I}-\vec{Q})$ is small. We do this by choosing the error $\varepsilon = 2^{-r'}$ allowed in our application of \autoref{lem:trans-prob-approx} exponentially smaller than the required one, i.e. $r$, to compensate for the increase in the error due to the multiplication of entries from $\vec{I}-\vec{Q}$ to obtain each entry. This can be done without affecting our complexity analysis since we can encode $r'$ using at most polynomially more bits than what was used to encode $r$. (See also the error-propagation analysis in the proof of \autoref{lem:acyclic-bc-finite} where we did almost the same). Now, since the right multiplication by $\vec{1}$ to get $\vec{N}\vec{1}$ can also be dealt with similarly, the only remaining complication is the error accumulated by division in computing $\nicefrac{\adj(\vec{I}-\vec{Q})}{\det(\vec{I}-\vec{Q})}$. Write $x$ and $y$ for the approximated numerator and denominator of any entry of the matrix $\nicefrac{\adj(\vec{I}-\vec{Q})}{\det(\vec{I}-\vec{Q})}$ and $a$ and $b$ for the actual values. Assuming we choose $r'$ as indicated in \autoref{rem:small-rel-error} so that $0 \leq \nicefrac{\varepsilon}{a}, \nicefrac{\varepsilon}{b} \leq 1$, the following hold:
\begin{align*}
    \frac{x}{y} ={} & \frac{a \pm \varepsilon}{b \pm \varepsilon} = \frac{a}{b} \left(\frac{1 \pm \frac{\varepsilon}{a}}{1 \pm \frac{\varepsilon}{b}}\right)\\
    {}={} & \frac{a}{b} \left(1 \pm \frac{\varepsilon}{a}\right)\left(\sum_{i=0}^\infty \left(\pm \frac{\varepsilon}{b}\right)^i \right) & \text{since } 0 \leq \frac{\varepsilon}{b} \leq 1\\
    {}={} & \left(\frac{a}{b} \pm \frac{\varepsilon}{b}\right)\left(\sum_{i=0}^\infty \left(\pm \frac{\varepsilon}{b}\right)^i \right)\\
    {}={} & \frac{a}{b} \pm \frac{\varepsilon}{b} +
    \left(\frac{a \pm \varepsilon}{b}\right)\left(\sum_{i=1}^\infty \left(\pm \frac{\varepsilon}{b}\right)^i \right)\\
    {}={} & \frac{a}{b} \pm \frac{\varepsilon}{b} +
    \left(\frac{a \pm \varepsilon}{b}\right)\left(\frac{1}{1 \pm \frac{\varepsilon}{b}} \right)\left(\frac{\varepsilon}{b}\right) & \text{again, as } 0 \leq \frac{\varepsilon}{b} \leq 1\\
    {}={} & \frac{a}{b} \pm \frac{\varepsilon}{b} +
    \left(\frac{a \pm \varepsilon}{b \pm \varepsilon}\right)\left(\frac{\varepsilon}{b}\right)
\end{align*}
and therefore, $|\nicefrac{x}{y} - \nicefrac{a}{b}|$ is the absolute value of the last expression above minus $\nicefrac{a}{b}$. We thus get the following inequalities.
\begin{align*}
    \left| \frac{x}{y} - \frac{a}{b} \right| = {} & \left| \frac{\varepsilon}{b} +
    \left(\frac{a \pm \varepsilon}{b \pm \varepsilon}\right)\left(\frac{\varepsilon}{b}\right) \right|\\
    {} \leq {} & \frac{\varepsilon}{b} +
    \left(\frac{a + \varepsilon}{b - \varepsilon}\right)\left(\frac{\varepsilon}{b}\right) & \text{because } 0 \leq \frac{\varepsilon}{a}, \frac{\varepsilon}{b} \leq 1\\
    {} = {} & \frac{\varepsilon}{b}\left(1 +
    \left(\frac{a + \varepsilon}{b - \varepsilon}\right)\right) 
 = \frac{\varepsilon}{b} \left( \frac{b+a}{b-\varepsilon} \right)\\
    {} = {} &
    \frac{\varepsilon}{b} \left(\frac{b+a}{b-\varepsilon}\right) \leq 
    \frac{\varepsilon}{b}\left(\frac{2}{b-\varepsilon}\right) & \text{as } 0 \leq a,b \leq 1
\end{align*}
We can now establish a lower bound for $b$ much like that of \autoref{eqn:lower-bnd-prob}. In symbols:
\[
    b \geq 2^{-2^{n^c}}
\]
holds, for some constant $c \in \mathbb{N}$. The last inequality above can now be rewritten as follows.
\begin{align*}
    \left| \frac{x}{y} - \frac{a}{b} \right| \leq {} & \frac{\varepsilon}{b}\left(\frac{2}{b-\varepsilon}\right) \leq \frac{2^{-r'}}{2^{-2^{n^c}}}\left(\frac{2}{2^{-2^{n^c}} - 2^{-r'}}\right)\\
    {}={} & \frac{2^{-r'}}{2^{-2^{n^c}}}\left(\frac{2 \left(2^{r' +2^{n^c}}\right)}{2^{r'} - 2^{2^{n^c}}}\right) = \frac{1}{2^{-2^{n^c}}}\left(\frac{2^{1 +2^{n^c}}}{2^{r'} - 2^{2^{n^c}}}\right)\\
    {} \leq {} & \frac{2^{3\left(2^{n^c}\right)}}{2^{r'} - 2^{2^{n^c}}}
\end{align*}
This means that if $r' \geq r + 4(2^{n^c})$ then our approximation has an absolute error of at most $2^{-r}$, exactly as required. Since this can be achieved using polynomially more bits than that needed to encode $r$, our complexity result holds too.
\end{proof}

\section{Time to termination for SIR models in the Blum-Shub-Smale model}\label{sec:sir-algo}
This section is meant as a more pragmatic approach to computing the expected time to termination. For this, we focus on the concrete SIR binomial chain. Furthermore, we change from our familiar Turing-machine model of computation to a Blum-Shub-Smale (BSS) machine~\cite{DBLP:conf/focs/BlumSS88}. In a BSS machine, we have registers that can hold arbitrary real numbers (including irrational ones) and applying rational operations on them takes a single time step. Finally, we also assume that $\exp(-h\beta)$ and $\exp(-h\gamma)$ are given as part of the input. We will show that, in this context, the expected time to termination can be computed in time polynomial with respect to the $1$-norm of the initial vector.

Let us write $Q$ for the set of states of the Markov chain induced by the given SIR binomial chain $(\vec{v},\vec{T})$, so $Q \subseteq \{0,1,\dots,N\}^3$, where $N = \norm{1}{\vec{v}}$, and the components of a state $\vec{m} \in Q$ correspond to susceptible, infectious, and recovered, in that order. We also write $A\subseteq Q$ for its absorbing states, thus $A = \{\vec{m} \in Q \mid m_2 = 0\}$ by \autoref{lem:succ}. Note that, since SIR binomial chains are closed, one of the components of the vector states is redundant and $|Q| \leq N^2$. Hence, the dimensions of $\vec{P}$ (the matrix representation of the transition probability function of the induced Markov chain) are at most $N^2$.

\begin{rem}[A first algorithm, polynomial in N]
    Recall that the vector $\vec{k^A}$ of expected hitting times can be computed using the expression in \autoref{thm:norris-hit}. This already gives us a naive algorithm that runs in time polynomial in $N$ if $\vec{P}$ is given.
    Indeed, Gaussian elimination requires only a cubic number of arithmetic operations. Therefore, given $\vec{P}$, we can compute $\vec{N} = (\vec{I} - \vec{Q})$ in time $O(N^6)$.
\end{rem}

\subsection{A second algorithm, now without Gauss}
In light of the proof of \autoref{thm:eoe-pspace}, one may wonder whether Gaussian elimination is needed to obtain $\vec{N}$. As we will briefly show, it can in fact be avoided in favor of back substitution because here too $\vec{P}$ and thus $\vec{I} - \vec{P}$ are upper triangular.

\begin{thm}\label{thm:n4-algo}
    Given an SIR BC $(\vec{T},\vec{v})$, we can compute $\vec{k^A}$ in time $O(N^4)$ in the Blum-Shub-Smale model.
\end{thm}

\noindent
Before proving the result, we observe that
\autoref{lem:succ} gives us a characterization of the nonzero terms in $\vec{P}$. In this way, for SIR binomial chains, the expression from \autoref{thm:norris-hit} is reduced to what is shown in the next lemma.
 
\begin{lem}\label{lem:belleq}
    For all states $\vec{m} \in Q$, we have that:
    \[
         k^A_{\vec{m}} = \begin{cases}
            0 & \text{if } m_2 = 0,\\
            1 + \sum_{n_1 = 0}^{m_1} \sum_{n_3 = m_3}^{m_2+m_3}
            P(\vec{m},\vec{n}) k^A_{\vec{n}} & \text{otherwise,}
        \end{cases}
    \]
    where $\vec{n} = (n_1,N-n_1-n_3,n_3)$.
\end{lem}

Our present goal is to order the (vector) states according to a total order as we did in \autoref{sec:ex-hit-abs}. The result above suggests a possible order to achieve this. In words, we will use the colexicographic ordering. Since SIR binomial chains are closed, the strict version of the order is defined as follows.
\[
(m_1,m_2,m_3) <_{\mathrm{colex}} (n_1,n_2,n_3) \iff m_3 < n_3 \text{ or } m_3=n_3 \text{ and } n_1 < m_1.
\]
We can now sort $Q$ based on $<_{\mathrm{colex}}$ so that $\vec{P}$ is in its canonical form (see \autoref{eqn:canon-form}). By \autoref{thm:abs-hit}, we have the following relation.
\begin{equation}
  (\vec{I} - \vec{Q})\vec{k^A}=\vec{1}  
\end{equation}
Hence, due to our choice of order on the states, together with \autoref{lem:succ}, we get that $\vec{I}-\vec{Q}$ is upper triangular.
This enables the computation of $\vec{k^A}$
by means of back substitution, which requires only a quadratic number of arithmetic operations.
\begin{lem}\label{lem:n4-assume}
    Given an SIR BC $(\vec{v},\vec{T})$, we can compute $\vec{k^A}$ in time $O(N^4)$ in the Blum-Shub-Smale model, assuming $\vec{P}$ is given.
\end{lem}
\noindent 
It remains to argue that $\vec{P}$ can be precomputed in $O(N^4)$. \revone{R1.3}{Unfortunately, computing one single entry of $\vec{P}$ using} \autoref{eq:Xpmf} \revone{R1.3}{seems to require time $O(N)$, meaning that a naive enumeration of pairs of states $(m_1,m_3)$ and $(n_1,n_3)$ with $P((m_1,N-m_1-m_3),(n_1,N-n_1-n_3,n_3))$ computed for each of them results in a total of $O(N^5)$.} Indeed, while exponentiation can be realized using a logarithmic number of operations via iterated squaring, we are not aware of an algorithm to compute factorials using a sublinear number of multiplications.\footnote{We do know that, in the Turing-machine model of computation, one can obtain seemingly better bounds~\cite{borwein85}. However, those depend on the complexity of multiplying numbers being a function of the amount of bits used to encode them.}

\subsection{Precomputation of the transition probabilities using dynamic programming}
Recall Equation \eqref{eq:Xpmf} gives us that for all
states $\vec{m} \neq \vec{n}$ such that $n_1 \leq m_1$ and $m_3 \leq n_3 \leq
m_2 + m_3$ the probability $P(\vec{m},\vec{n})$ of transitioning from $\vec{m}$ to $\vec{n}$ is as
follows.
\begin{align}
  & \binom{m_1}{m_1-n_1}\exp(-h\beta m_2)^{n_1}
  (1-\exp(-h\beta m_2))^{m_1-n_1} \label{eqn:final-si}\\
  & \binom{m_2}{n_3-m_3}(1-\exp(-h\gamma))^{n_3-m_3}\exp(-h\gamma)^{m_2-n_3+m_3} \label{eqn:final-ir}
\end{align}
The crux of our algorithm is the following recurrence.
\begin{lem}\label{lem:recrel}
  Let $\vec{m},\vec{n}$ be states such that $n_1 < m_1$ and $m_3 < n_3 \leq m_2 + m_3$. Then, \(
    P(\vec{m},\vec{n}) = \alpha(\vec{m},\vec{n}) P(m_1-1,m_2,m_3+1,\vec{n}),
  \) where:
  \[
      \alpha(\vec{m},\vec{n}) = 
      \frac{m_1(m_2-n_3+m_3+1)(1-\exp(-h\gamma))(1-\exp(-h\beta m_2))}{(n_3-m_3)(m_1-n_1)\exp(-h\gamma)}.
  \]
\end{lem}
\begin{proof}
  Note that $P(\vec{m},\vec{n}), P(m_1-1,m_2,m_3+1,\vec{n})> 0$ by
  \autoref{lem:succ} and our assumptions.
  We first consider the binomial coefficients of \eqref{eqn:final-si} and \eqref{eqn:final-ir}.
  \begin{align*}
    & \binom{m_1}{m_1-n_1} \binom{m_2}{n_3-m_3}\\
    {}={} & \frac{m_1(m_2-n_3+m_3+1)}{(n_3-m_3)(m_1-n_1)} \binom{m_1-1}{m_1-1-n_1}
    \binom{m_2}{n_3-(m_3+1)}
  \end{align*}
  Importantly, since we have assumed that $m_3 < n_3$ and $n_1 < m_1$, the denominator of the fraction above is not $0$.
  On the other hand, for the terms involving ``probabilities'' --- that is, exponential terms --- we observe that the following is true for \eqref{eqn:final-si}.
  \begin{align*}
  & \exp(-h\beta m_2)^{n_1}
  (1-\exp(-h\beta m_2))^{m_1-n_1}\\ {}={} &(1-\exp(-h\beta m_2)) \exp(-h\beta m_2)^{n_1}
  (1-\exp(-h\beta m_2))^{m_1-1-n_1}
  \end{align*}
  Similarly, for the probabilities in \eqref{eqn:final-ir} we note the following.
  \begin{align*}
    & (1-\exp(-h\gamma))^{n_3-m_3}\exp(-h\gamma)^{m_2-n_3+m_3}\\
    {}={} & \frac{1-\exp(-h\gamma)}{\exp(-h\gamma)} (1-\exp(-h\gamma))^{n_3-(m_3+1)}\exp(-h\gamma)^{m_2-n_3+m_3+1}
  \end{align*}
  Indeed, $\alpha(\vec{m},\vec{n})$ is exactly the product of the coefficients on the right-hand sides of the above equalities. This concludes the proof as the product of the left-hand sides is exactly $P(\vec{m},\vec{n})$ and that of the right-hand sides (after factoring out $\alpha$) is $P(m_1-1,m_2,m_3+1,\vec{n}$).
\end{proof}

Based on \autoref{lem:recrel}, we propose the following algorithm.

\begin{algo}
    An efficient algorithm to compute $P(\vec{m},\vec{n})$ for all $\vec{m}$ and $\vec{n}$
    \label{alg:probs}
    \begin{algorithmic}[1]
        \For{$m_1 = 0, \dots,
        N$}
            \State $X(\vec{m},\vec{m}) = 1$, with $\vec{m} = (m_1,0,N-m_1)$
        \EndFor
        \For{$M = 1, \dots, N$}
          \For{$m_2 = 1, \dots, M$}
              \State $\vec{m} \gets (M - m_2,m_2,N-M)$
              \For{$n_3 = m_3, \dots,m_2+m_3$} \Comment{Fix $n_1=m_1$}
                \State $X(\vec{m},\vec{n}) \gets P(\vec{m},\vec{n})$, with $\vec{n}=(m_1,N-m_1-n_3,n_3)$\label{loc:fixn1}
              \EndFor
              \For{$n_1 = 0, \dots,m_1$} \Comment{Fix $n_3=m_3$}
                \State $X(\vec{m},\vec{n}) \gets P(\vec{m},\vec{n})$, with $\vec{n}=(n_1,N-n_1-m_3,m_3)$\label{loc:fixn3}
              \EndFor
              \For{$n_3 = m_3+1, \dots, m_2+m_3$}\label{loc:reccomp-start}
                  \For{$n_1 = 0, \dots, m_1-1$}\label{loc:reccomp-n1}
                     \State $\vec{n} \gets (n_1,N-n_1-n_3,n_3)$
                     \State $\vec{m'} \gets (m_1-1,m_2,m_3+1)$
                     \State $X(\vec{m},\vec{n}) \gets \alpha(\vec{m},\vec{n})X(\vec{m'},\vec{n})$
                  \EndFor
              \EndFor\label{loc:reccomp-end}
          \EndFor
      \EndFor
    \end{algorithmic}
\end{algo}

We observe that \autoref{alg:probs} clearly terminates because all for-loops are bounded and there are no jump statements in the code. Regarding its time complexity, we note that we only ever nest at most $4$ for-loops and $P(\vec{m},\vec{n})$ is never used (i.e. computed naively in $O(N)$ steps) within a for-loop nesting of depth $4$.
\begin{prop}[Complexity]
    The worst-case time complexity of \autoref{alg:probs} is $O(N^4)$ in the Blum-Shub-Smale model.
\end{prop}
\noindent 
For correctness, we have the following statement.
\begin{prop}[Correctess]
    Let $X$ be as computed by \autoref{alg:probs}. Then, $X(\vec{m}, \vec{n}) = P(\vec{m},\vec{n})$ for all pairs of states $\vec{m}, \vec{n}$ such that $n_1 \leq m_1$ and $m_3 \leq n_3 \leq m_2 + m_3$.
\end{prop}
\begin{proof}
    First, one needs to take care that the previous transition probabilities must be defined before they are being used. Note that this is only relevant in lines \ref{loc:reccomp-start}--\ref{loc:reccomp-end}. It thus follows from the order in which the states $\vec{m}$ are traversed that $X(\vec{m'},\vec{n})$ has already been computed in a previous iteration.
    
    Now, it is easy to see that $X(\vec{m},\vec{n}) = P(\vec{m},\vec{n})$ if either $m_1 = n_1$ or $m_3 = n_3$ since lines \ref{loc:fixn1} and \ref{loc:fixn3} literally assign the right-hand side to the left-hand side of the equality. If neither equality holds, then $X(\vec{m},\vec{n})$ is computed in the loop \ref{loc:reccomp-start}--\ref{loc:reccomp-end} and $X(\vec{m},\vec{n}) = P(\vec{m},\vec{n})$ by \autoref{lem:recrel}. Indeed, the conditions that $n_1 < m_1$ and $m_3 < n_3$ are guaranteed by values over which $n_3$ and $n_1$ range in the loops from lines \ref{loc:reccomp-start} and \ref{loc:reccomp-n1}.
\end{proof}

\autoref{thm:n4-algo} follows from the results above, together with \autoref{lem:n4-assume}.

\section{Time to termination using probabilistic model checkers}\label{sec:encoding}
In this section, we take a step back from SIR binomial chains and move to the general class of acyclic BCs. It seems difficult to generalize our algorithms for SIR binomial chains, especially since we exploited the property of them being closed. Instead, we propose to encode the computation of the time to termination as an instance of a \emph{probabilistic model checking} problem. In particular, we aim at a small encoding (i.e. without explicitly constructing the induced Markov chain) that does not depend on the exact initial vector state. This is desirable to avoid reencoding the original BC every time the initial population changes. We target the PRISM input
language~\cite{DBLP:conf/cav/KwiatkowskaNP11} for our encoding since it is supported by state-of-the-art tools such as PRISM itself, Storm~\cite{DBLP:journals/sttt/HenselJKQV22}, and Modest~\cite{DBLP:conf/tacas/HartmannsH14}.

Let us fix an acyclic BC $(\vec{v},\vec{T})$.
As in the previous section, we will avoid the problem of irrational probabilities by asking that they be part of the input. More precisely, we will assume that for all linear functions $T_{ij} = \sum_{\ell=1}^k a_\ell x_\ell + a_0$ from $\vec{T}$ we have precomputed the values $p_{ij\ell} = e^{-a_\ell}$. These are then stored using some finite precision and taken to be part of the input. Using the new notation, we can rewrite \autoref{eqn:prob-formula} as follows.
\begin{equation}\label{eqn:prob-form2}
    \sum_{\vec{M} \in \wit{\vec{u}}{\vec{w}}} 
    \prod_{(i,j) \in \supp{\vec{T}}} \binom{u_i}{M_{ij}} \left(1 - p_{ij0} \prod_{\ell=1}^k p_{ij\ell}^{u_\ell}\right)^{M_{ij}} \left(p_{ij0}\prod_{\ell=1}^k p_{ij\ell}^{u_\ell}\right)^{u_i - M_{ij}}
\end{equation}

In the next section, instead of focusing on the exact syntax of the PRISM language, we introduce a new model we call \emph{stochastic counter machines}. Encoding such machines into a PRISM program is straightforward for someone familiar with the language. Additionally, this intermediate reduction will allow us to highlight the difficulties of attempting a direct encoding of binomial chains into PRISM.

\subsection{Stochastic counter machines}
A stochastic counter machine (or SCM) in dimension $k$ is a tuple $\mathcal{M} = (Q, q_0,\vec{c_0}, T)$ where $Q$ is a finite set of (control) states, $q_0 \in Q$ is the initial state, $\vec{c_0} \in \mathbb{N}^k$ is the initial value of the \emph{counters} (to which we sometimes refer using Greek letters), and $T$ is a finite set of tuples $(q,g,u,q',p)$ where: $q,q' \in Q$, the \emph{guard} $g$ is a matrix-vector pair $\vec{A} \in \mathbb{Q}^{n \times d},\vec{b} \in \mathbb{Q}^n$, the \revtwo{R2.7}{\emph{update}} $u$ is also a matrix-vector pair $\vec{U} \in \mathbb{Q}^{d \times d}, \vec{r} \in \mathbb{Q}^d$, and $p \in \mathbb{Q}_{\geq 0}$.

A \emph{configuration} of the SCM is a pair $(q,\vec{c}) \in Q \times \mathbb{N}^k$ --- we often write $q(\vec{c})$ instead of $(q,\vec{c})$.
The \emph{transition} $(q,g,u,q',p)$, with $g = (\vec{A},\vec{c})$ and $u = (\vec{U},\vec{r})$, is \emph{enabled} from a configuration $q(\vec{c})$ if and 
only if $\vec{A}\vec{c} \leq \vec{b}$. Further, the \emph{$t$-successor} of $q(\vec{c})$ is the configuration $q'(\vec{U}\vec{c} + \vec{r})$ with probability $p$. 

We need some constraints on SCMs to ensure configuration-successor pairs are unique and that a proper distribution is defined. For all configurations $q(\vec{c})$, we ask that:
\begin{itemize}
    \item For all transitions $t_1, t_2$, if both $t_1$ and $t_2$ are enabled from $q(\vec{c})$ and the $t_1$- and $t_2$-successors of $q(\vec{c})$ are the same then $t_1 = t_2$.
    \item Let $t_1, t_2, \dots$, where $t_i = (q,g_i,u_i,q'_i,p_i)$, be an enumeration of the transitions enabled from $q(\vec{c})$. Then, we have that $\sum_{i\geq 1} p_i = 1$.
\end{itemize}
\noindent 
The semantics of the SCM $\mathcal{M}$ is given, like that of a BC, via an induced Markov chain $\mathcal{C}_{\mathcal{M}} = (S,s_0,P)$ where $S = Q \times \mathbb{N}^k$ and $s_0 = q_0(\vec{c_0})$. For the transition probability function $P$, we have that $P(q(\vec{c}),q'(\vec{c'})) = p > 0$ if and only if $q'(\vec{c'})$ is the $t$-successor of $q(\vec{c})$, for some transition $t$, with probability $p$. The constraints above ensure that $P$ is well defined.

\begin{figure}[t]
    \centering
    \begin{tikzpicture}[initial text={},inner sep=0.5pt,minimum size=0mm,every
      node/.style={font=\small},every state/.style={minimum size=0mm}]
        \node[state, initial](c0) {$q_0$} ;
        \node[state,right=3cm of c0](c1){$q_1$};
        \node[state,right=2cm of c1](c2){$q_2$};

        \node[above= of c1, circle, minimum size=0.2cm, fill] (dot){};
        \path[-] (c1)  edge  node [right] {$[\chi_3 > 0]$} (dot);
        \draw [->] (dot) to[out=180, in=160, looseness=4, edge node={node [left] 
            {$1-p: \begin{pmatrix*}[l] \chi_1 &  \leftarrow \chi_1 \\ \chi_2 &  \leftarrow \chi_2 \\ \chi_3 & \leftarrow \chi_3-1 \end{pmatrix*}$}}] (c1);

        \draw [->] (dot) to[out=0, in=20, looseness=4, edge node={node [right] {$p: \begin{pmatrix*}[l] \chi_1 &  \leftarrow \chi_1\\ \chi_2 &  \leftarrow \chi_2 + 1\\ \chi_3 &  \leftarrow \chi_3 - 1\end{pmatrix*}$}}] (c1);
        
        \path[->,auto]
        (c0) edge node[below,]{$1:\begin{pmatrix*}[l] \chi_1 \leftarrow \chi_1 \\ \chi_2 \leftarrow 0 \\ \chi_3 \leftarrow \chi_1 \end{pmatrix*}$} (c1)
        (c1) edge node[below]{$[\chi_3 = 0]$} (c2)
        (c2) edge[loop right] (c2)
        ;
    \end{tikzpicture}
    \caption{A stochastic counter machine that simulates a binomial distribution.
    Control states are depicted as circles and transitions $(q,g,u,q',p)$ as arrows from $q$ to $q'$. Guards are shown in square brackets along a transition, together with probability-update pairs $p : ( u )$ with $p$ being the probability and $u$ the update. If the guard is trivial (i.e. it is satisfied by all configurations) then we omit it; if the update is the identity, we also omit it; and if the probability is $1$, again we omit it. Finally, we group transitions from a common state and with the same guard using a solidly filled circle with the guard shown before the circle and the individual probabilities and updates shown after.}
    \label{fig:bingad1}
\end{figure}
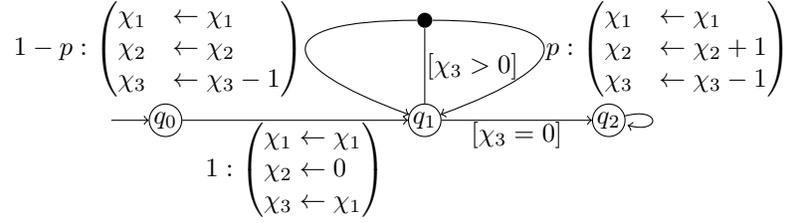

\subsection{Binomial and Bernoulli stochastic counter machines}

The SCM depicted in 
\autoref{fig:bingad1}, encodes a
binomial distribution with success probability $p$. Intuitively, from the initial configuration $q_0(c_1,c_2,c_3)$ the machine simulates $c_1$ Bernoulli trials, each with success probability $p$. The third counter $\chi_3$ is used to count from $c_1$ (copied when transitioning from $q_0$ to $q_1$) to $0$ by decrementing on both transitions from $q_1$ to itself. The right transition from $q_1$ to itself simulates a success and thus increments the value of the second counter, where we store the total number of successful Bernoulli trials. The left transition leaves the value of the second counter untouched since it corresponds to a failed Bernoulli trial.

From the discussion above, it should be clear that the SCM implements a binomial distribution in the sense that the distribution over successor configurations with control state $q_2$ is given exactly by the probability mass function of that distribution.
\begin{lem} \label{lem:bingadget1}
    Let $\mathcal{M}$ be the SCM from \autoref{fig:bingad1} and consider its induced Markov chain $\mathcal{C}_{\mathcal{M}}$. For all $t,t', c_1,c_2,c_3 \in \mathbb{N}$ with $t' - t \geq c_1$ it holds that: 
    \[
    \Pr(X_{t'} = q_2(c_1,m,0) \mid X_t = q_0(c_1,c_2,c_3)) = B(m;c_1,p).
    \]
\end{lem}
\noindent 
The reader may have already started realizing our intention: We want to use binomial SCMs to simulate the innnermost products from \autoref{eqn:prob-form2} which correspond to binomial distributions. Unfortunately, unless $p_{ij\ell}$ is $1$ for all $\ell > 0$, the SCM we just presented is not good enough because it simulates Bernoulli trials with a constant success probability $p$. To resolve this, we present a second SCM that
allows us to model a Bernoulli trial where the success probability is a function of the counter values of the initial configuration. More precisely, the function is $1 - p_0\prod_{\ell = 1}^{k} p_\ell^{c_k}$, where $p_\ell \in \mathbb{Q}_{\geq 0}$ and $p_\ell \leq 1$ for all $0 \leq \ell \leq k$. 

The SCM depicted in \autoref{fig:gad3}, encodes a \emph{parametric} Bernoulli distribution. From the initial configuration $q_0(c_1,c_2,c_3,c_4)$, and much like in the binomial SCM, the last counter is used as a temporary counter while the second-to-last counter holds the result of the trial: $c_3$ if it fails and $c_3 + 1$ if it succeeds. For instance, on the transition from $q_0$ to $q_1$, the last counter copies the value $c_1$. Then, while its value is nonzero, there is a transition back to $q_1$ with probability $p_1$ that decrements the counter. If its value has reached zero, the machine transitions to $q_2$ while setting the value of the last counter to $c_2$. There, it again loops on the same state and decrements the counter while its value is nonzero, this time with probability $p_2$. On both loops, the alternative transitions reach state $q_x$. Clearly, the probability of eventually reaching the state $q_y$ is $p_0\prod_{\ell = 1}^{2} p_\ell^{c_\ell}$. Hence, the probability of eventually reaching $q_x$ is $1 - p_0\prod_{\ell = 1}^{2} p_\ell^{c_\ell}$. Now, for $\ell > 2$, the SCM can be extended further below $q_2$ by adding copies of the same structure \emph{mutatis mutandis}.

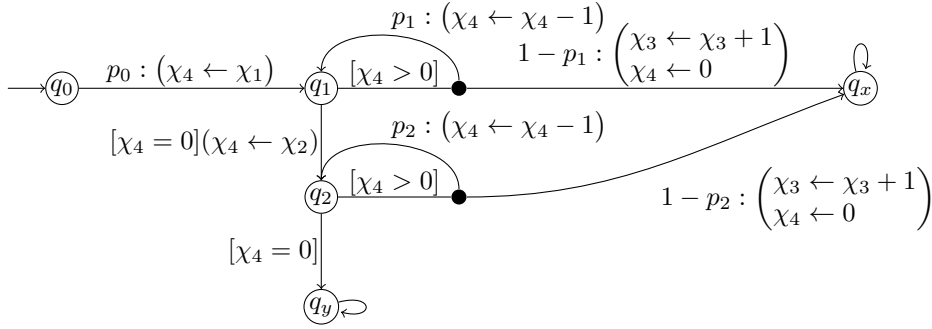
\begin{figure}
    \centering
    \begin{tikzpicture}[initial text={},inner sep=0.5pt,minimum size=0mm,every
      node/.style={font=\small},every state/.style={minimum size=0mm}]
        \node[state, initial](q0){$q_0$};
        \node[state, right=3cm of q0] (q1){$q_1$};
        \node[state,below=1cm of q1](q2){$q_2$};
        \node[right=1.5cm of q1, circle, minimum size=0.2cm, fill] (dot1){};
        \node[right=1.5cm of q2, circle, minimum size=0.2cm, fill](dot2){};
        \node[state, right=5cm of dot1](qx){$q_x$};
        \node[state,below=1cm of q2](qy){$q_y$};

      \draw [->,auto]
      (q0)  edge node {$p_0:\begin{pmatrix*}[l] \chi_4 \leftarrow \chi_1 \end{pmatrix*}$} (q1)
      (qx) edge[loop above] (qx)
      (qy) edge[loop right] (qy);
      
      \draw [-,auto] (q1) edge node {$[\chi_4 >0]$} (dot1)
      (q2) edge node {$[\chi_4 > 0]$} (dot2);
      
      \draw [->,auto]
      (dot1) edge[out=90, in=90] node[swap] {$p_1:\begin{pmatrix*}[l] \chi_4 \leftarrow \chi_4 -1  \end{pmatrix*}$} (q1)
      (dot1) edge node{$1-p_1:\begin{pmatrix*}[l] \chi_3 \leftarrow \chi_3 + 1 \\ \chi_4 \leftarrow 0  \end{pmatrix*}$} (qx)
      (q1)  edge node[swap] {$[\chi_4 = 0] (\chi_4 \leftarrow \chi_2) $} (q2)
      (dot2) edge[out=0,in=200] node[swap]{$1-p_2:\begin{pmatrix*}[l] \chi_3 \leftarrow \chi_3 + 1 \\ \chi_4 \leftarrow 0  \end{pmatrix*}$} (qx)
      (q2) edge node[swap]{$[\chi_4 = 0]$} (qy)
      (dot2) edge[out=90, in=90] node[swap] {$p_2:\begin{pmatrix*}[l] \chi_4 \leftarrow \chi_4 -1  \end{pmatrix*}$} (q2)
      ;
    
    \end{tikzpicture}
    \caption{A stochastic counter machine that simulates a parametric Bernoulli distribution with $\ell = 2$. Most notation is as in \autoref{fig:bingad1}, but we also omit trivial updates.}
    \label{fig:gad3}
\end{figure}

\begin{lem}\label{lem:gadget2}
    Let $\mathcal{M}$ be an SCM as in \autoref{fig:gad3} for $\ell \in \mathbb{N}$ and consider its induced Markov chain $\mathcal{C}_{\mathcal{M}}$. For all $\vec{c} \in \mathbb{N}^{k+2}$ and all $t,t' \in \mathbb{N}$ with $t' - t \geq \norm{1}{\vec{c}}$ it hold that:
    \begin{align*}
    & 1 - \Pr(X_{t'} = q_x(c_1,\dots,c_k,c_{k+1}+1,0) \mid X_t = q_0(\vec{c}))\\
    {}={} &  \Pr(X_{t'} = q_y(c_1,\dots,c_k,c_{k+1},0) \mid X_t = q_0(\vec{c}))
    = 
    p_0\prod_{\ell = 1}^{k} p_\ell^{c_\ell}.
    \end{align*}
\end{lem}

\begin{figure}[t]
    \centering
    \begin{tikzpicture}[scale = 0.8, initial text={},inner sep=0.5pt,minimum size=0mm,every
      node/.style={font=\small},every state/.style={minimum size=0mm}]
        \node[state, initial above](q0){$q_0$};
        \node[state, right=3cm of q0] (q1){$q_1$};
        \node[state, right=3cm of q1](q3){$q_3$};
        \node[state, above= of q1] (q2) {$q_2$};
        \node[above= of q2, circle, minimum size=0.2cm, fill] (dot1){};

        \node[state, below= 1.5cm of q1] (q4) {$q_4$};
        \node[state, below= 1.5cm of q0] (q5) {$q_5$};
        \node[below= of q4, circle, minimum size=0.2cm, fill] (dot2){};

        \draw [->] (q0)  edge node [below] {$1:\begin{pmatrix} \alpha_0 \leftarrow s  \end{pmatrix}$} (q1);
        \draw [->] (q1)  edge node [right] {$[\alpha_0 > 0] : \begin{pmatrix*}[l]
            \alpha_1 \leftarrow i \\ \alpha_0 \leftarrow \alpha_0 -1
        \end{pmatrix*}$} (q2);
        \draw [->] (q1)  edge node [below] {$[\alpha_0 = 0]$} (q3);

        \draw[-] (q2) edge node [right] {$[\alpha_1>0]$} (dot1);

        \draw [->] (q2) to[out=200, in=160, looseness=2, edge node={node [left] {$[\alpha_1 = 0]$}}] (q1);

        \draw [->] (dot1) to[out=180, in=170, looseness=4, edge node={node [left] {$1-p_0:\begin{pmatrix}
            \chi_0 \leftarrow \chi_0 + 1
        \end{pmatrix}$}}] (q1);

        \draw [->] (dot1) to[out=0, in=0, looseness=4, edge node={node [right] {$p_0:\begin{pmatrix}
            \alpha_1 \leftarrow \alpha_1 - 1
        \end{pmatrix}$}}] (q2);

        \draw [->] (q3) to[out=270, in=0, looseness=1, edge node={node [pos=0.3,right=0.2cm] {$1:\begin{pmatrix}
            \alpha_1 \leftarrow i
        \end{pmatrix}$}}] (q4);

        \draw [->] (q4)  edge node [below] {$[\alpha_1 = 0] $} (q5);
        
        \draw [->] (q5)  edge node [left] {$ 1: \begin{pmatrix*}[l]
            s \leftarrow s - \chi_0 \\ 
            i \leftarrow i + \chi_0 -\chi_1 \\ 
            r \leftarrow r + \chi_1 \\
            \chi_0 \leftarrow 0 \\
            \chi_1 \leftarrow 0 \\
        \end{pmatrix*}$} (q0);

        \draw[-] (q4) edge node [right] {$[\alpha_1 > 0]$} (dot2);

        \draw [->] (dot2) to[out=0, in=-20, looseness=4, edge node={node [right] {$p:\begin{pmatrix}
            \chi_1 \leftarrow \chi_1 + 1\\
            \alpha_1 \leftarrow \alpha_1 - 1
        \end{pmatrix}$}}] (q4);

        \draw [->] (dot2) to[out=180, in=200, looseness=4, edge node={node [left] {$1-p_1:\begin{pmatrix}
            \alpha_1 \leftarrow \alpha_1 - 1
        \end{pmatrix}$}}] (q4);
    \end{tikzpicture}
    \caption{The SCM we construct for the SIR binomial chain from \autoref{sec:sir}}
    \label{fig:scmsir}
\end{figure}
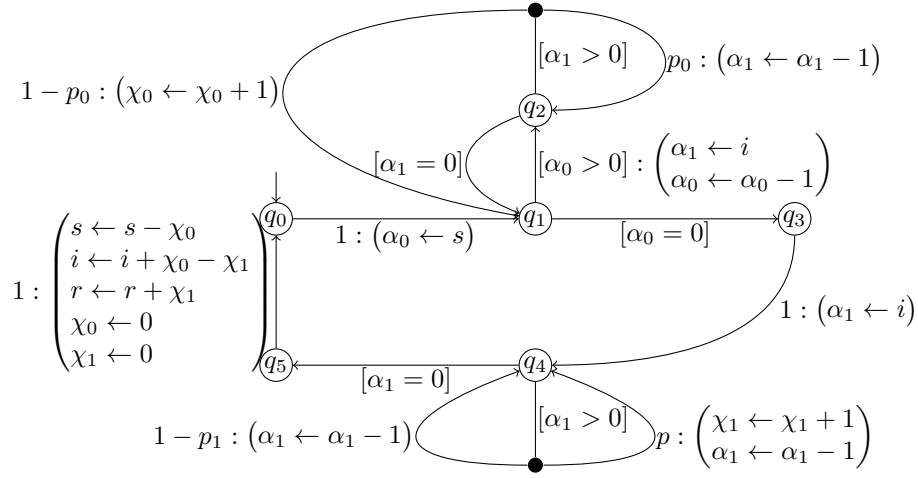

\subsection{The final construction and going beyond reachability}
We can suitably compose parametric Bernoulli SCMs with and binomial SCMs to simulate a BC. After a series of binomial SCMs (to simulate the innermost product from \autoref{eqn:prob-form2}), we have at most $k^2$ counters whose values correspond to the number of individuals that are to be transferred for this simulated transition of the BC.
As an example, we give the full construction for SIR binomial chains in
\autoref{fig:scmsir}. The SCM
uses $7$ counters, to which we will refer as $s, i, r, \alpha_0, \alpha_1, \chi_0, \chi_1$.
Let $p_0 = 1 - \exp(-h\beta)$ and $p_1 = 1 - \exp(-h\gamma)$ and observe that $\beta_0 \sim B(s, p_0^i)$ and $\beta_1 \sim B(i,
p_1)$. as expected. Hence, going once from $q_0$ to itself simulates exactly one
timestep in the binomial chain.

So far, the formal claims establishing correctness of our construction speak only about transition and eventual reachability probabilities.
We can also encode the expected
time to termination by adding \emph{rewards} to the transitions of the SCM. To be precise, we can formalize rewards as pairs $(g,r)$ of guards and rational rewards. Then, all transitions taken from a configuration $q(\vec{c})$ that satisfies the guard yield a reward of $r$.\footnote{To the expert reader: Yes, this is exactly a reward structure in the PRISM language.} For our purpose, the reward function should be $+1$
for each simulated step of the binomial chain before reaching an absorbing state, e.g. on the transition from $q_0$ to $q_1$, and $0$
for all other transitions. To make sure only transitions before reaching the absorbing states get a reward, we can add one counter that keeps track of the sum of all counter values corresponding to components of the BC which influence the value of some $T_{ij}$. Then, this sum will be $0$ if and only if the current (simulated) vector state is absorbing.  

The next claim summarizes the properties of the constructed SCM. There, we subindex probability functions to highlight the probability space of the (induced) Markov chain to which they belong. Furthermore, we refer to the total number of states and counters of an SCM as its \emph{size}.

\begin{thm}
    For all acyclic BC $\mathcal{B}$ we can construct an SCM $\mathcal{M}$ with rewards, of size $k^{O(1)}$, and an injection $\mu$ from states of $\mathcal{B}$ to configurations of $\mathcal{M}$ so that the following hold.
    \begin{itemize}
    \item There exists $m \in \mathbb{N}$ such that for all $t,t',n \in \mathbb{N}$, with $t < t'$ and $m < n$, and all states $\vec{u},\vec{w}$ of $\mathcal{B}$ we have
    \(
    \Pr_{\mathcal{B}}(X_{t'} = \vec{w} \mid X_t = \vec{w}) = \Pr_{\mathcal{M}}(X_{t'+ n} = \mu(\vec{w}) \mid X_t = \mu(\vec{u})).
    \)
    \item Moreover, the expected time to reach $\vec{w}$ from $\vec{u}$ in $\mathcal{C}_{\mathcal{B}}$ is the expected sum of rewards before reaching $\mu(\vec{w})$ from $\mu(\vec{u})$ in $\mathcal{C}_{\mathcal{M}}$.
    \end{itemize}
\end{thm}

\section{Experimental Results}\label{sec:experiments}

We implemented our translation from binomial chains to stochastic counter machines in a prototype tool we called \emph{Inform}.
More concretely, Inform translates files in an explicit representation of the transfer matrix of a binomial chain into into a PRISM-language file.
%
Using it and a probabilistic model checker such as \emph{Storm}~\cite{DBLP:journals/sttt/HenselJKQV22}, we can verify properties of the encoded binomial chain. The second model is a simplified COVID-19 model based on
\cite{ABRAMS2021100449}. In this section we do precisely this. Below, we introduce the binomial chain we studied and the properties we checked.

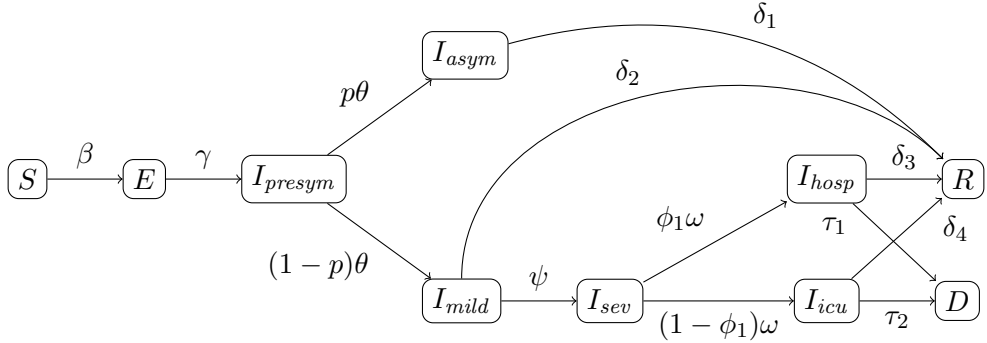
\begin{figure}[t]
    \centering
    \begin{tikzpicture}
    \node[rectangle,rounded corners,draw](s){$S$};
    \node[rectangle,rounded corners,draw,right= of s](e){$E$};
    \node[rectangle,rounded corners,draw,right= of e](ip){$I_{\mathit{presym}}$};
    \node[rectangle,rounded corners,draw,above right= of ip](ia){$I_{\mathit{asym}}$};
    \node[rectangle,rounded corners,draw,below right= of ip](im){$I_{\mathit{mild}}$};
    \node[rectangle,rounded corners,draw,right= of im](is){$I_{\mathit{sev}}$};
    \node[rectangle,rounded corners,draw,right=2cm of is](ii){$I_{\mathit{icu}}$};
    \node[rectangle,rounded corners,draw,above= of ii](ih){$I_{\mathit{hosp}}$};
    \node[rectangle,rounded corners,draw,right= of ih](r){$R$};
    \node[rectangle,rounded corners,draw,right= of ii](d){$D$};
    \path[->,auto]
      (s) edge node{$\beta$} (e)
      (e) edge node{$\gamma$} (ip)
      (ip) edge node{$p\theta$} (ia)
      (ip) edge node[swap]{$(1-p)\theta$} (im)
      (ia) edge[out=15] node{$\delta_1$} (r)
      (im) edge[out=90] node{$\delta_2$} (r)
      (im) edge node{$\psi$} (is)
      (is) edge node{$\phi_1\omega$} (ih)
      (is) edge node[swap]{$(1-\phi_1)\omega$} (ii)
      (ih) edge node{$\delta_3$} (r)
      (ih) edge node[pos=0.05,swap]{$\tau_1$} (d)
      (ii) edge node[swap]{$\tau_2$} (d)
      (ii) edge node[swap,pos=0.9]{$\delta_4$} (r)
      ;
  \end{tikzpicture}
    \caption{Overview of the flow of individuals in the (single-age group version of the) binomial chain for the early stages of COVID-19 in Belgium from \cite{ABRAMS2021100449}}
    \label{fig:sircovid}.
\end{figure}

\subsection{The Belgian COVID-19 binomial chain}
\autoref{fig:sircovid} shows a graph representation of the transfer matrix of the binomial chain we study. It comes from a compartmental model based on the classical SEIR model where we add to SIR an intermediate compartment for individuals that are \emph{exposed} but not yet infectious. Furthermore, infectious individuals are partitioned into those who are: $I_{\mathit{presym}}$, presymptomatic infectious individuals; $I_{\mathit{asym}}$, asymptomatic infectious; $I_{\mathit{mild}}$, mild cases; $I_{\mathit{sev}}$, severe cases; $I_{\mathit{hosp}}$, hospitalized; and $I_{\mathit{icu}}$, hospitalized and in the intensive care unit. Now, all compartments are further split into age groups. The individuals who can infect susceptible ones are only those who are presymptomatic, asymptomatic, mild cases, or severe cases --- intuitively, those who are hospitalized are isolated and are therefore assumed to not have contact with susceptible individuals. The contact rates between infectious and susceptible individuals, per age pair, are given by \emph{contact matrices} fitted by the authors of~\cite{ABRAMS2021100449} from national statistics. 

All of the above results in the linear function $T_{SE}$, i.e. the entry of the transfer matrix $\vec{T}$ corresponding to the indices for the compartments $S$ and $S$, being nontrivial. In fact, all other entries of $\vec{T}$ are constants while $T_{SE}$ depends on around $~400$ components of the current vector state. (We refer the interested reader to \cite{ABRAMS2021100449} for the ODEs and the binomial-chain formulation of the time-discretized compartmental model and to its supplementary materials for the values of all constants.) For our experiments, we simplified the COVID-19 binomial chain to summarize all age groups back into a single compartment.

Note that the binomial chain, is acyclic, but not closed. This seems more like a quirk of the modelling formalism rather than a desirable feature in the context of the time bounds studied by the authors of~\cite{ABRAMS2021100449} which span part of the COVID-19 epidemic and some months after it. This, paired with the observation that binomial chains seem to allow the possibility of spontaneous infection of all susceptible individuals in a single transition, inspired our choice of properties to check for in the model.
\begin{description}
    \item[PopInc] How likely is it that the population does not remain constant?
    \item[OS] How likely is it that a given portion of the population moves from one compartment to another, in one shot?
    \item[EoE] What is the expected time before the end of the epidemic?
\end{description}
For the OS property, we chose to focus on when this happens along the very first transition of the binomial chain. We give our exact encoding of these properties as given to the model checkers in \autoref{tab:props}.

\subsection{Experimental setup}
Since the COVID-19 model is very large, besides Storm, we also used a statistical model checker called
\emph{Modest} \cite{DBLP:conf/tacas/HartmannsH14}. Statistical model checkers
usually scale better at the price of only providing 
confidence intervals
instead of the exact probability value with which a given property holds. Since our main objective was to probe how model checkers scale on larger instances of the binomial chain, we used increasingly larger populations and an initial vector state having all compartments empty except for $S$, $I_{\mathit{asym}}$, $I_{\mathit{mild}}$, and $I_{\mathit{sev}}$.

All experiments were run on a cluster where each node had
an Intel(R) Xeon(R) Platinum 8168 CPU @ 2.70GHz with 64GiB of memory and no
GPU.
%
%
We ran Storm from the \emph{movesrwth/storm:stable} docker container and, based on a number of local experiments, chose the sparse engine for all numbers reported henceforth. Storm's version was 1.7.1. For Modest, we used version v3.1.237-g2f62162c7. We denote time-outs with TO: we stopped the computation after 1 hour (for small populations); memory-outs with MO: the program was terminated because it ran out of memory. All code and scripts to re-run the experiments can be found here: \url{https://github.com/UA-FOTS/inform}.

\begin{table}[t]
    \centering
    \caption{The properties of interest stated in PRISM-style probabilistic computation tree logic and in terms of the Belgian COVID-19 model. The PopInc property states: What is the probability (P=?) that we avoid an error state (that is only reached when the population is not preserved) until the infectious compartments are depleted.
When these are empty, it is impossible for more people to become infected. Note that for the OS property, we on the start of the simulation: What is the probability that we avoid infecting susceptible individuals until we infect all of them?}
    \begin{tabular}{ll} 
        \toprule
        Property & Probabilistic Temporal Logic Formula  \\
        \midrule
        PopInc & P=? [$(q \neq \mathit{error})$ U $(E + I_{presym} + \dots + I_{icu} = 0)$] \\
        OS & P=? [$(S \geq S\_\mathit{init})$ U $(E = S\_\mathit{init} + E\_\mathit{init})$] \\
        EoE & R\{time\_step\}=? [F $(E + I_{presym} + \dots + I_{icu} = 0)$]\\
        \bottomrule
        \end{tabular}
    
    \label{tab:props}
\end{table}

\subsection{Results}
For both the population increase (PopInc) as well as the end of epidemic (EoE) property, we see Storm running out of memory already for populations of 10. The one-shot (OS) property performs really well in Storm. This may be because our formalization of
the OS property only checks the first time individuals change compartment. 
Storm seems to be taking this into account when building the state-space. 
Modest performs significantly better than Storm for the PopInc property. Moreover, the size of the population has a much smaller impact on the run-time compared to Storm. However, for the EoE property, we see that Modest struggles. In order to not time out for small instances, the \emph{width of the confidence interval} of Modest was set to $0.9$ and, even in this case, runtimes were significantly higher than for the PopInc property.

\begin{table}[t]
    \centering
    \caption{Performance of Storm compared to Modest on the COVID-19 model. For Storm, we used default parameters and the sparse engine. For the EoE property, Modest was run with max run length 0 and width 0.9; for the others, with max run length 0 and width 0.01.}
    \begin{tabular}{lccc} 
        \toprule
        Property & Population & Storm & Modest \\
         & $(S,I_{\mathit{asym}},I_{\mathit{mild}},I_{\mathit{sev}})$ & runtime & runtime \\
        \midrule
  PopInc & (2, 1, 1, 1) & 2.916 &  3.4s \\
         & (3, 1, 1, 1) & 21.201 &  6.9s \\
         & (4, 1, 1, 1) & 297.23s & 5.9s  \\
         & (5, 1, 1, 1) & 2352.066s & 5.7s  \\
         & (6, 1, 1, 1) & 14756.769s & 4.5s  \\
         & (7, 1, 1, 1) & MO & 4.3s  \\
         \midrule
         EoE & (2, 1, 1, 1) & 3.400s & 1023.5s  \\
         & (3, 1, 1, 1) & 27.314s &  997.8s \\
         & (4, 1, 1, 1) & 570.862s &  1069.3s \\
         & (5, 1, 1, 1) & 5325.083s &  1048.3s \\
         & (6, 1, 1, 1) & 42751.95 &  1039.0s \\
         & (7, 1, 1, 1) & MO &  1080.8s \\
         \midrule
      OS & (2, 1, 1, 1) & 0.123s & 0.2s  \\
         & (3, 1, 1, 1) & 0.125s & 0.1s  \\
         & (4, 1, 1, 1) & 0.141s & 0.2s  \\
         & (5, 1, 1, 1) & 0.186s & 0.1s  \\
         & (6, 1, 1, 1) & 0.195s & 0.1s  \\
         & (7, 1, 1, 1) & 0.223s & 0.1s  \\
         & (8, 1, 1, 1) & 0.244s & 0.2s  \\
         & (9, 1, 1, 1) & 0.274s & 0.1s  \\
         & (10, 1, 1, 1) & 0.306s & 0.2s  \\
         \bottomrule
    \end{tabular}
    \label{tab:modeststormperformance}
\end{table}

Finally, and based on the results summarized above, we used Modest to analyze the COVID-19 model with realistic populations. The results are shown in  \autoref{fig:popincultimate}. We observe that the run-time of Modest grows almost linearly with respect to $S$.

\begin{figure}[ht]
    \centering
    \begin{tikzpicture}[scale=0.8]
    \begin{axis}[%
        xlabel = {$(S, 1, 1, 1)$ population},
        ylabel = {seconds},
        scatter/classes={%
        a={draw=black}}]
        \addplot[scatter,only marks,%
            scatter src=explicit symbolic]%
        table[meta=label] {
            x y label
            10  6.1 a  \\
            20  9.3 a \\
            40  11.7 a \\
            80  12.1 a \\
            160  25.2 a \\
            320  27.7 a \\
            640  68.1 a \\
            1280  74.7 a \\
            2560  257.1 a \\
            5120  504.0 a \\
            10240 689.8 a \\
            20480 1869.1 a \\
            40960 2774.2 a \\
            81920 6420.6 a \\
            163840 8882.6 a \\
            327680 29609.8 a \\
        };
    \end{axis}
\end{tikzpicture}
    \caption{Modest runtimes for large populations and the PopInc property. Here, the width was the default $0.01$ and the populations were $(S,I_{\mathit{asym}},I_{\mathit{mild}},I_{\mathit{sev}}) = (S, 1, 1, 1)$ with $S$ is increasing.}
    \label{fig:popincultimate}
\end{figure}
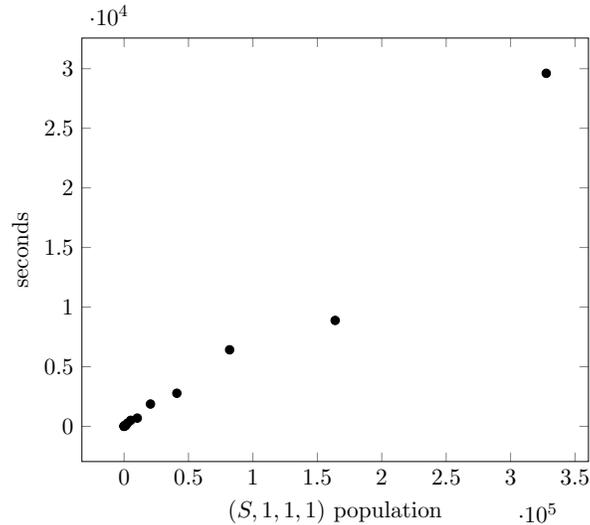

It is worth mentioning that we did compare the translation-based solution against our custom-optimized algorithm from \autoref{sec:sir-algo} for SIR binomial chains and the EoE property. Even for small populations, our algorithm from \autoref{sec:sir-algo} outperformed both model checkers.

\section{Conclusion and future work}
We started the study of binomial chains through the lens of formal methods. In this work, we provided two main theoretical results. First, we established that acyclic binomial chains almost surely terminate. Second, we proved the the problem of approximating the time to termination is in \textbf{PSPACE} and gave a direct algorithm for the exact problem (ignoring the complexity of arithmetic and the annoyances of irrational numbers). 
\revtwo{R2.1}{Unfortunately, extending the algorithm from} \autoref{sec:sir-algo} to general binomial chains seems hard. The key result enabling \autoref{alg:probs} to compute transition probabilities in $O(N^4)$ was \autoref{lem:recrel} and it made use of the explicit formula for transition probabilites we had manually derived for SIR binomial chains. While similar results could be obtained by hand for fixed transfer matrices, it is unclear to us what a (meta)result for general binomial chains would look like.

We also provided a more pragmatic approach in the form of an encoding into the PRISM language.
For this last approach we also presented some experiments. Based on the empirical results, we can conclude that state-of-the-art probabilistic model checkers are not (yet) powerful enough to deal with epidemiological models like the one proposed in \cite{ABRAMS2021100449}. Indeed, while Modest is capable of handling simple probabilistic queries for realistic populations, it still seems to struggle with quantitative queries such as the expected end of epidemic property. In this direction, more research is needed to find, for instance, good abstractions.

Finally, there are natural decision problems for binomial chains that we did not consider in this work. For instance, the following, based on values studied in~\cite{ABRAMS2021100449}, seem relevant.
\begin{description}
    \item[Finite-horizon peak] asks to compute the maximal expected population in a given compartment within a finite horizon (given in binary).
    \item[Finite-horizon accumulated population] asks to compute the total expected population having been in a given compartment within a finite horizon (given in binary).
\end{description}

\bibliographystyle{alphaurl}
\bibliography{refs}

\end{document}